\documentclass[12pt]{article}
\usepackage{makeidx}         % allows index generation
\usepackage{multicol}        % used for the two-column index
\usepackage[bottom]{footmisc}% places footnotes at page bottom
\usepackage{amsmath, amssymb, amsthm}
\usepackage{psfrag,epsf}
\usepackage{enumerate}
\usepackage{url} % not crucial - just used below for the URL 
\usepackage{hyperref}
\usepackage{booktabs}
\usepackage{xcolor}
\usepackage[authoryear]{natbib}
\usepackage[english]{babel}
\usepackage{float}
\usepackage[final]{graphicx}
\usepackage{subfig}
\usepackage{enumitem}
\usepackage{multirow}
\usepackage{arydshln}

\usepackage{authblk}
\usepackage{geometry} %[left=0.5cm,right=0.5cm,top=2.5cm,bottom=2cm]
\usepackage{lscape}
\usepackage{rotating}
\usepackage{changes}
\definechangesauthor[color=red]{LM}

\hypersetup{
			colorlinks=true,
			linkcolor=black,
            linktoc=page,
			anchorcolor=black,
			citecolor=blue,
			urlcolor=blue,
}

\def \qmo{``}

\def \qmcsp{'' }
\def \bs{\mathbf}
\def \bold{\boldsymbol}

\newtheorem{theorem}{Theorem}
\newtheorem{remark}{Remark}
\newcommand\omicron{o}
\newcommand{\red}{\textcolor{black}}
\newcommand{\redd}{\textcolor{black}}

\def \bqmatrix{\begin{pmatrix}}
\def \eqmatrix{\end{pmatrix}}

\graphicspath{:/media/luca/749403CC94038FB2/LUCA/Dottorato/Progetto di Dottorato/Multivariate Quantile regression/Code/}

\newcommand{\blind}{1} % blind: 0

\begin{document}

\def\spacingset#1{\renewcommand{\baselinestretch}%
{#1}\small\normalsize} \spacingset{1}

\if1\blind
{
  \title{\bf Unified unconditional regression for multivariate quantiles, M-quantiles and expectiles}
  \author{\,\\
  Luca Merlo\\
    Department of Human Sciences, European University of Rome\\
    and \\
    Lea Petrella \\
    MEMOTEF Department, Sapienza University of Rome\\
    and \\
    Nicola Salvati \\
    Department of Economics and Management, University of Pisa\\
    and \\
    Nikos Tzavidis \\
    Department of Social Statistics and Demography, Southampton Statistical Sciences Research Institute, University of Southampton}
  \maketitle
} \fi

\if0\blind
{
  \bigskip
  \bigskip
  \bigskip
  \begin{center}
    {\LARGE\bf Unified unconditional regression for multivariate quantiles, M-quantiles and expectiles}
\end{center}
  \medskip
} \fi

%\title*{Univariate and Multivariate Unconditional M-quantile regression}
%\titlerunning{ Univariate and Multivariate Unconditional M-quantile regression} 
%%\subtitle{\emph{Modello Hidden Markov Multivariato ad effetti misti per la stima congiunta di quantili condizionati}} % Contribution Title in Italian
%\author{Merlo Luca, Petrella Lea, Salvati Nicola and Tzavidis Nikos}
%% \authorrunning{Short Title} % Use for an abbreviated version of your contribution title if the original one is too long
%\institute{Merlo Luca \at Department of Statistical Sciences, Sapienza University of Rome, 
%\email{luca.merlo@uniroma1.it}
%\and Petrella Lea \at MEMOTEF Department, Sapienza University of Rome,
%\email{lea.petrella@uniroma1.it}
%\and Salvati Nicola \at Department of Economics and Management, University of Pisa,
%\email{salvati@ec.unipi.it}
%\and Tzavidis Nikos \at Department of Social Statistics and Demography, University of Southampton,
%\email{N.TZAVIDIS@soton.ac.uk}}
%\maketitle

\bigskip
\begin{abstract}
In this paper, we develop a unified regression approach to model unconditional quantiles, M-quantiles and expectiles of multivariate dependent variables exploiting the multidimensional Huber's function.
%In this paper we develop a generalized quantile regression to model unconditional quantiles, expectiles and M-quantiles of multivariate dependent variables in a unifying linear regression setting.
%Our approach is based on the multidimensional extension of Huber's M-quantile, which encompasses multivariate quantiles and expectiles. 
 To assess the impact of changes in the covariates across the entire unconditional distribution of the responses, we extend the work of \cite{firpo2009unconditional} by running a mean regression of the recentered influence function on the explanatory variables.
 %We extend the work of \cite{firpo2009unconditional} to assess the impact of changes in the covariates on the unconditional distribution of the response by \red{using Recentered Influence Function regressions} %running a mean regression of the recentered influence function of the unconditional quantiles, M-quantiles and expectiles
 %on the explanatory variables 
 We discuss the estimation procedure and establish the asymptotic properties of the derived estimators. A data-driven procedure is also presented to select the tuning constant of the Huber's function.
 %and derive the asymptotic properties of the corresponding estimator. %We discuss the estimation of the model and derive the asymptotic properties of the estimators.
 %By using the multidimensional extension of Huber's function, our approach offers a unifying framework for modeling and estimating multivariate unconditional quantiles, expectiles and M-quantiles in a linear regression setting.} 
 The validity of the proposed methodology is explored with simulation studies and through an application using the Survey of Household Income and Wealth 2016 conducted by the Bank of Italy.
\end{abstract}

\noindent%
{\it Keywords:}  Influence Function, M-estimation, Multivariate Data, RIF Regression, Unconditional Partial Effect
\vfill
%Multivariate M-quantile
%Multivariate M-quantiles
%Robust estimation, inference

\newpage
\spacingset{1.9} % DON'T change the spacing!
\section{Introduction}\label{sec:intro}
%The methodology has been applied in a number of different settings where the object of interest is the unconditional distribution of outcomes.
%Several methodologies (Machado and Mata 2005; Melly 2005) use conditional quantiles regressions as primary tools to infer entire distributions and counterfactual distributions even when the object of interest is the unconditional quantiles.

%researcher is interested in estimating how the whole distribution of the dependent variable responds to changes in the covariates,}
%effect of unions on wage inequality would yield a misleading answer
When researchers wish to determine the effect of relevant predictors across the entire distribution of the dependent variable of interest, Quantile Regression (QR), as introduced by \cite{koenker1978regression}, plays a crucial role in providing a much more complete statistical analysis %about the relationship between the response and explanatory variables 
 compared to the classical mean regression. Indeed, it allows to model conditional quantiles of a response as a function of explanatory variables and it has been greatly exploited for the study of non-Gaussian, heavy-tailed and highly skewed data. For a detailed survey and list of references of the most used QR techniques, % with applications to survival analysis, treatment effects and longitudinal data, among others
 please refer to \cite{koenker2005quantile} and \cite{koenker2017handbook}. 

 %Though the well-known QR proposed by \cite{koenker1978regression} has proven to be a powerful tool to explore conditional distributions, %unlike conditional means, however, conditional quantiles do not average up to their unconditional population counterparts. 
 However, if one is interested in how the whole unconditional distribution of the outcome responds to changes in the covariates, QR methods would yield misleading inferences (see \citealt{firpo2009unconditional, borah2013highlighting, maclean2014application}). %killewald2014motherhood
 As explained by \cite{frolich2013unconditional} in a simple example relating wages to years of education, the unconditional 90-th quantile refers to the high wage workers, whereas the 90-th quantile conditional on education refers to the high wage workers within each education class, who however may not necessarily be high earners overall. Presuming a strong positive correlation between education and wages, it may well be that the 90-th quantile among high school dropouts is lower than, say, the median of all Ph.D. graduates. The interpretation of the 90-th quantile is thus different for conditional and unconditional quantiles. From a policy perspective, while the welfare of highly educated people with relatively low wages catches little interest, the welfare of the poor, i.e., those located in the lower end of the unconditional distribution of wages, attracts a lot of attention in the political debate.
 
%in unconditional effects
 \red{In medicine, the analysis of epidural analgesia on the duration of the second stage of labor (see \citealt{zhang2012causal}) %(from full cervical dilation to delivery of the fetus) 
 gives another example where researchers are interested in changes in the quantiles, $q_\tau$, of the unconditional distribution of the response $Y$, $F_Y (y)$. Obstetricians are particularly concerned about the unconditional effect $d q_\tau (\mathfrak{p}) / d\mathfrak{p}$ of increasing the proportion of patients receiving epidural analgesia, $\mathfrak{p} = Pr[X=1]$, on the $\tau$-th quantile of the unconditional distribution of the duration of second-stage labor, where $X = 1$ if the patient receives the treatment (epidural) and $X = 0$ otherwise. Unfortunately, the coefficient $\beta_\tau$ from a conditional quantile regression, $\beta_\tau = F^{-1}_Y (\tau \mid X = 1) - F^{-1}_Y (\tau \mid X = 0)$, is generally different from $d q_\tau (\mathfrak{p}) / d\mathfrak{p} = (\mbox{Pr}[Y > q_\tau \mid X = 1] - \mbox{Pr}[Y > q_\tau \mid X = 0]) / f_Y (q_\tau)$, the effect of increasing the proportion of patients receiving epidural analgesia on the $\tau$-th quantile of the unconditional distribution of $Y$.}
 There are, indeed, numerous applications of practical relevance where the ultimate research objective is the unconditional distribution of the dependent variable, % the goal is to infer about the influence of specific covariates on the unconditional distribution of the response, 
 as in the context of the earnings disparities between different groups of workers, the effect of education on earnings, or the distributional impacts of a particular treatment in a given population (\citealt{borah2013highlighting, frolich2013unconditional, huffman2017equality, firpo2018decomposing}).
  Therefore, a number of proposals has been introduced in the literature to estimate these unconditional effects (\citealt{machado2005counterfactual, melly2005decomposition}). %gosling2000changing chernozhukov2013inference
%Therefore, to answer such research objective, a number of proposals has been introduced in the literature to estimate the effects of covariates on the unconditional distribution of interest
%model the unconditional quantiles of a response variable in a regression framework 
 
 %\textcolor{red}{\cite{gosling2000changing} and \cite{machado2005counterfactual} proposed semiparametric procedures to estimate the entire unconditional distribution of the dependent variable in the presence of covariates. Both approaches are based on the estimation of the conditional distribution by quantile regression; then, the former averages the conditional distribution function with respect to the covariates, while the latter resorts to resampling procedures to yield the unconditional distribution function of the response.} %The two approaches differ in the method used to construct the unconditional distributions implied by the conditional model. %\cite{melly2005decomposition} further extends their works by solving the problem of crossing of different quantile curves and by determining the asymptotic distribution of the proposed estimator. 
%Although the UQR target parameter represents the effect on outcomes in the overall target population, CQR target parameters represent effects on outcomes in specific parts of the target population defined by the conditioning implemented by CQR.
 Motivated by this interest, \cite{firpo2009unconditional} proposed the Unconditional Quantile Regression (UQR) approach for \red{estimating the impact of changes in the distribution of the explanatory variables on quantiles of the unconditional distribution of a dependent variable.} %modeling unconditional quantiles of a dependent variable as a function of the explanatory variables.
  %The UQR provides more policy-relevant information than conventional QR when one is primarily concerned with the estimation of the overall effect of policy interventions on the entire unconditional distribution of a response variable of interest in a given population. 
 This method builds upon the concept of Recentered Influence Function (RIF) which originates from a widely used tool in robust statistics, namely the Influence Function (IF, see \citealt{hampel1974influence, huber2009robust}). %Thoroughly discussed in \cite{hampel1974influence} and \cite{huber1981robust}, the IF assesses the sensitivity of a distributional statistic of interest to a small contamination in the population distribution function. 
 %The RIF, in particular, may be thought as the contribution of an individual observation to a given distributional statistic. 
 In particular, the UQR of \cite{firpo2009unconditional} consists of regressing %can be implemented by regressing %developed a regression method to evaluate the effect of changes in the distribution of the covariates on unconditional quantiles, denoted in their work as the Unconditional Quantile Partial Effect (UQPE). To pursue this goal, the authors regress 
 the RIF of the unconditional quantile of the outcome variable on the explanatory variables using either Ordinary Least Squares (OLS), logistic regression or the nonparametric regression of \cite{newey1994asymptotic}. %\red{The authors then show that this approach can be used to estimate the effect on the unconditional quantile of a small location shift in the distribution of covariates, holding everything else constant.} %average derivative of the RIF corresponds to 

 Their approach represents an important contribution on quantile regression methods and its validity is also demonstrated by the growing scientific literature spanning from medicine, economics, social inequalities and agriculture.\\
 %Their approach represents an important contribution on quantile regression methods and its validity is also demonstrated by the growing scientific literature spanning from medicine (\citealt{borah2013highlighting}), economics (\citealt{murakami2019spatially, dong2020asymmetric}), social inequalities (\citealt{rodriguez2016unconditional, huffman2017equality}) and agriculture (\citealt{mishra2015impact, bonanno2018food}).\\
 %%%%%%%%%%%%%%%%%%%%%%%%%%%%%%%%%%%%%%%%%%%%%%%%%%
 %to the traditional conditional QR of \cite{koenker1978regression} 
 %concerned with the estimation of qua
  %The UQR method has been positively considered in a number of different fields, like 
  %is receiving increasing attention from a policy perspective
%taxation (\citealt{maclean2014application}), welfare inequality (\citealt{agyire2018unconditional}), food policy (\citealt{mishra2015impact} and \citealt{bonanno2018food}) and hedonic analysis (\citealt{murakami2019spatially}).
 These studies, however, focus on a univariate regression framework. \red{In the analysis of multivariate data, univariate approaches are not appropriate for this purpose, as they provide only partial pictures of the phenomenon under investigation. In these cases, the research interest may focus not only on a regression model for each outcome, but also on accounting for the dependence structure between the responses.} When the problem under investigation involves multivariate dependent variables, the method of \cite{firpo2009unconditional} cannot be easily extended to higher dimensions due to the non existence of a natural ordering in a $p$-dimensional space, $p>1$ (see \citealt{serfling2002quantile, kong2012quantile, koenker2017handbook, petrella2019joint, merlo2022marginal}). %merlo2021quantile
%(see \cite{chakraborty2003multivariate, hallin2010multivariate, kong2012quantile, koenker2017handbook, Stolfietal, chavas2018multivariate, alfo2020m, charlier2020multiple}
 
%Originally  for univariate responses, of multivariate response variables
With this paper, we contribute to the current literature extending the univariate UQR approach of \cite{firpo2009unconditional} to a more general multivariate setting.  Particularly, we propose to employ the multidimensional Huber's function defined in \cite{hampel2011robust} to build a unified unconditional regression approach that encompasses multivariate quantiles, M-quantiles (\citealt{breckling1988m}) and expectiles (\citealt{newey1987asymmetric}).
\\
 %encompassing multivariate unconditional quantiles, M-quantiles (\cite{breckling1988m}) and expectiles (\citealt{newey1987asymmetric}) based on the multidimensional Huber influence function (\citealt{hampel2011robust}) approach.\\ 
In the statistical literature, the Huber's function has been used to define the M-quantile for robust modeling of the entire distribution of univariate response variables, extending the ideas of M-estimation of \cite{huber1964} and \cite{huber2009robust}.
% by introducing a class of asymmetric influence functions to model the entire distribution of a response variable of interest. 
 This method provides a ``quantile-like'' generalization of the mean based on influence functions that combines in a common framework the robustness and efficiency properties of quantiles and expectiles, depending on the choice of the Huber's tuning constant; \red{the latter offering higher estimation efficiency and computational advantages compared with the former when there are no outliers in the data.} In the multivariate framework, the multidimensional Huber's function (\citealt{hampel2011robust}) %is widely used \red{for constructing robust M-estimators of multivariate location in scenarios with heavy-tailed variables and correlated outliers (see \citealt{maronna1976robust, peker2016fitting}).} In particular, it 
 has been exploited by \cite{breckling1988m} %considered the multidimensional generalization of the univariate Huber's loss function
%Depending on the type of influence function used, M-quantiles may reduce to standard quantiles or expectiles. 
 to define the multivariate M-quantile using a direction vector in the Euclidean $p$-dimensional space in order to establish a suitable ordering procedure for multivariate observations. 
%The multivariate M-quantile along a specified direction is obtained by minimizing the multidimensional Huber loss function (\citealt{hampel2011robust}), which utilises a tuning constant to control the influence of outliers on estimation.
 %that can adjust the robustness of the estimator
 Subsequently, \cite{kokic2002new} generalized their definition by introducing a class of multivariate M-quantiles based on weighted estimating equations, which includes multivariate quantiles and expectiles depending on the value of the tuning constant. \red{This proposal provides a robust technique for summarizing the distribution of multidimensional data and overcomes the shortcomings of the definitions of multivariate geometric or spatial quantiles (\citealt{chaudhuri1996geometric}) and expectiles (\citealt{herrmann2018multivariate}) when extremes observations are of interest; see \cite{girard2017intriguing}.}
%When the interest of the research is on the entire conditional distribution, in addition to the classical quantile regression, a possible alternative approach is to consider the M-quantile regression proposed by \cite{breckling1988m}. 

In our paper, we rely on the \cite{kokic2002new} approach using the multidimensional Huber's function to model 
%the multivariate M-quantile and exploit the multidimensional Huber influence function in \cite{hampel2011robust} as a tool to model 
unconditional quantiles, M-quantiles and expectiles of multivariate response variables in a unified regression framework, by choosing the tuning constant appropriately. %Specifically, if the tuning constant %of the multidimensional Huber function 
\\
 In order to analyze the impact of changes in the distribution of explanatory variables on the entire unconditional distribution of the responses, following \cite{firpo2009unconditional}, we regress the RIF of the proposed model on the covariates, producing the Unconditional Quantile, M-Quantile and Expectile Partial Effect (UQPE, UMQPE, UEPE) according to the selected value of the tuning constant. From the theoretical point of view, we %We estimate these partial effects by running a mean regression of the RIF on the set of covariates and 
 establish the asymptotic properties of the corresponding estimators using the Bahadur representation (\citealt{bahadur1966note}). %For univariate outcome variables, our method conveniently reduces to the univariate definition of M-quantile which in turn includes the quantile and expectile as special cases. 
 Furthermore, we propose a data-driven method based on cross-validation for selecting the tuning constant that accounts for possible outliers in the data.

Using simulation studies, we illustrate the finite sample properties of the proposed methodology %and the performance of the data-driven tuning constant approach %, showing the validity and the robustness of our procedure 
 under different data generating processes. %to model unconditional quantiles, M-quantiles and expectiles of univariate and multivariate response variables.
 From an empirical standpoint, we demonstrate the usefulness of this method through the analysis of the Survey on Household Income and Wealth (SHIW) 2016 conducted by the Bank of Italy. %(\href{https://www.bancaditalia.it/statistiche/tematiche/indagini-famiglie-imprese/bilanci-famiglie/distribuzione-microdati/index.html}{\nolinkurl{https://www.bancaditalia.it}}).
 In particular, we fit the proposed model to evaluate the effect of economic and socio-demographic characteristics of Italian households on the unconditional distributions of family wealth and consumption, accounting both for the correlation between the outcomes and influential observations in the sample. %using the UMQR. Then, we develop an UMMQR to jointly study the unconditional distributions of household wealth and consumption as a function of the same predictors. 
  The proposed multivariate approach allows us to consider consumption and wealth as part of a collective framework and it can be of great interest to investigate \red{the unconditional effect of covariates on families' spending and wealth, with particular emphasis on those with jointly low or high consumption and wealth levels.} %of the responses' unconditional distribution.\\ % whether the effect of the covariates is more pronounced on low-quantiles (low-consumption and low-wealth families) than on high-quantiles (high-level spending and wealthy households) 
  
%for modeling unconditional M-quantiles
%low-consumption and low-wealth families (low-quantiles) than on high-level spending and wealthy households (high-quantiles).}\\
%what would happen to the overall population level

%In doing so, it highlights the importance of considering income, consumption and wealth together as part of a conceptual framework.
%It has also stressed the need for a multi-dimensional approach that considers these three dimensions together, in order to gain a more complete understanding of household economic well-being
   
  %The proposed approach for modeling unconditional M-quantiles allows to investigate whether the covariates have a different impact on household's consumption and wealth, while accounting for influential observations in the sample.\\
  
%for: (i) outlier-robust inference of the model parameters; (ii) the identification of the heterogeneous effect of covariates across the unconditional distribution of the responses that would go undetected one had used mean regression.\\
 %potentially

%it conveys a more comprehensive understanding of the heterogeneous impact of covariates across the unconditional distribution of the responses than mean regression.
%it may be preferable to mean regression when the aim is to characterize the heterogeneous impact of covariates across the distribution of the responses.

 The remainder of the paper is organized as follows. In Section \ref{sec:pre}, we revise the RIF and its properties. Section \ref{sec:ummq} introduces the proposed unconditional regression model for multivariate response variables and provides a detailed discussion of the asymptotic properties of the introduced estimators. Finally, the empirical application is presented in Section \ref{sec:app}, while Section \ref{sec:con} concludes. The simulation study and all the proofs are provided in the Supplementary Materials.
 %Section \ref{sec:sim} discusses the simulation study and the results
 %the main notation and revise the univariate quantile regression model

\section{Notation and preliminary results}\label{sec:pre}
%for analyzing other distributional statistics  
%it is possible to use RIFs of other distributional statistics of X as explanatory variables to better capture how changes in the distribution of X affect the distribution statistic v(F Y ).
%Preliminaries on the influence function and Univariate M-quantile Regression
In this section, we present the main notation and concepts which we use throughout the paper. Specifically, we review the notion of Recentered Influence Function (RIF) which originates from the Influence Function (IF) of \cite{hampel1974influence}. \red{Then, we present the Unconditional Partial Effect (UPE) %and the policy effect 
 introduced by \cite{firpo2009unconditional} that leads} us to analyze \red{the impact of changes in the distribution of covariates} on the unconditional distribution of the response variable.

%by reviewing the Recentered Influence Function (RIF) proposed by \cite{hampel1974influence} and the Unconditional Partial Effect (UPE) of \cite{firpo2009unconditional}.}\\
%To better explain the link between M-quantiles and the unconditional regression method of \cite{firpo2009unconditional}, we review the Recentered Influence Function (RIF) and introduce the key parameter of interest%in the estimation
%, namely the Unconditional Partial Effect (UPE).\\

%In this work, the object of interest is the effect of a small increase in the location of the distribution of the explanatory variable ${X}$ on the $\tau$-th, with $\tau \in (0,1)$, M-quantile of the unconditional distribution of $Y$. 
Let ${\bs Y}$ denote a \red{vector-valued} random variable belonging to an arbitrary sample space $\mathcal{Y}$, \red{which can be either a subset or equal to $\mathbb{R}^p$}, with absolutely continuous distribution function $F_{\bs Y}$ and consider a vector-valued functional $\nu (F_{\bs Y})$ where $\nu:\mathcal{F}_{\nu} \rightarrow \mathbb{R}^p$, such that $\mathcal{F}_{\nu}$ is the collection of all distributions on $\mathcal{Y}$ for which $\nu$ is defined. \red{The functional $\nu$ can belong to a wide class of distributional statistics. For example, $\nu$ can be a location parameter characterizing the distribution of $\bs Y$, a measure of scatter, as well as many inequality measurements such as concentration functions.}
%\red{The functional $\nu$ can belong to a wide class of distributional statistics. For example, $\nu$ can be a location parameter characterizing the distribution of $\bs Y$ such as the mean or a quantile, a measure of scatter such as the covariance matrix, as well many inequality measurements such as, quantile ratios, inter-quantile ranges, concentration functions, and the Gini coefficient.} %for instance
%a class of distribution functions such that $F_{\bs Y} \in \mathcal{F}_{\nu}$
%where the domain of T is the set of all distributions in $\mathcal{F} (\mathcal{Y})$ for which T is defined]
%where the domain of T is the collection of all distributions on $\mathcal{Y}$ for which T is defined]
%is a convex subset of the set of all finite signed measures on $\mathcal{Y}$
%defined on a suitable subset of the set of probability measures on $\mathcal{Y}$

The IF allows to study the effect of an infinitesimal contamination in the underlying distribution $F_{\bs Y}$ at a point $\bs y$ on the statistic $\nu (F_{\bs Y})$ we are interested in. Let us consider $\Delta_{\bs y}$ the probability measure that puts mass 1 at the value $\bs y$ and let $F_{\bs Y,t \Delta_{\bs y}} = (1 - t) F_{\bs Y} + t \Delta_{\bs y}$ represent the mixing distribution with $t \in [0, 1]$. Following \cite{hampel2011robust}, the $p$-dimensional IF of $\nu(F_{\bs Y})$ is defined as:
\begin{equation}\label{eq:IF}
IF(\bs y; \nu) = \lim_{t \rightarrow 0} \frac{\nu (F_{\bs Y,t \Delta_{\bs y}}) - \nu(F_{\bs Y})}{t}.
\end{equation}
%= \frac{\partial \nu (F_{\bs Y,t \Delta_{\bs y}})}{\partial t} \big|_{t=0}.
Using the definition of IF in \eqref{eq:IF}, \cite{firpo2009unconditional} considered the RIF to analyze the statistic $\nu (F_{\bs Y})$ after a perturbation of $F_{\bs Y}$ in the direction of $\Delta_{\bs y}$. In particular, the RIF is defined as the first two terms of the von Mises linear approximation (\citealt{mises1947asymptotic}) of the corresponding statistic $\nu (F_{\bs Y,t \Delta_{\bs y}})$ with $t = 1$, namely:
\begin{equation}\label{eq:RIF0}
RIF(\bs y; \nu) = \nu(F_{\bs Y}) + \int IF(\bs s; \nu) d \Delta_{\bs y} (\bs s) =  \nu(F_{\bs Y}) + IF(\bs y; \nu).
\end{equation}
\red{The RIF in \eqref{eq:RIF0} can be interpreted as a linear approximation to a possibly complex and nonlinear statistic $\nu(F_{\bs Y})$ measuring how it is affected by an infinitesimal perturbation in $F_{\bs Y}$.}
%a linear approximation measuring how an infinitesimal perturbation in $F_{\bs Y}$ affects a possibly complex and non-linear statistic $\nu(F_{\bs Y})$.

%Eq. \eqref{eq:RIF0} shows that the RIF is a linear approximation which measures how an infinitesimal perturbation in $F_Y$ affects a possibly complex and non-linear statistic $\nu(F_Y)$. %a linear approximation to the non-linear function of a distributional statistic

 \red{In the presence of a set of covariates $\bs X \in \mathcal{X}$, with $\mathcal{X} \subset \mathbb{R}^k$ being the support of $\bs X$, \cite{firpo2009unconditional} suggested the use of the RIF in \eqref{eq:RIF0} for analyzing the impact on $\nu(F_{\bs Y})$ due to changes in the distribution of $\bs X$, $F_\mathbf{X}$. In particular, in order to incorporate the effect of the explanatory variables, by the law of iterated expectations %the fact that the IF in \eqref{eq:IF} integrates up to zero, 
 it follows from \eqref{eq:RIF0} that:
\begin{equation}\label{eq:v0}
\nu(F_{\bs Y}) = \int RIF(\bs y; \nu) dF_{\bs Y} (\bs y) = \int \mathbb{E}[RIF(\bs Y; \nu) \mid \mathbf{X} = \bs {x}] dF_\mathbf{X} (\mathbf{x}),
\end{equation}
where in the first equality we used the fact that $\int IF(\bs y; \nu) d F_{\bs Y} (\bs y) = \bs 0$ (see \cite{hampel2011robust}, p. 226) and in the second one we substituted in $F_{\bs Y} (\bs y) = \int F_{\bs Y \mid \mathbf{X}} (\bs y \mid \mathbf{x}) dF_\mathbf{X} (\mathbf{x})$.}
%where $F_\mathbf{X}$ is the distribution function of $\bs X$.
\\
  From \eqref{eq:v0} it can be seen that when one is interested in the impact of a change in the covariates $\bs X$ on a specific distributional statistic $\nu (F_{\bs Y})$, the $\mathbb{E}[RIF(\bs Y; \nu) \mid {\bs X} = \bs {x}]$ can be modeled as a function of $\bs X$, which can be easily implemented using regression methods for the conditional mean (see \citealt{firpo2009unconditional, firpo2018decomposing} and \citealt{Rios2020}).
 \redd{More formally, in this work we are mainly interested in a small location shift $t$ in the distribution of covariates $\bs X$ from $F_{\bs X}$ to the distribution $G_{\bs X}$ of the $k$-dimensional vector $\tilde{\bs X}$, where $\tilde{X}_l = X_l$ for $l \neq j$, $l = 1,\dots,k$, and $\tilde{X}_j = X_j + t$. In this case, let $\boldsymbol{\alpha}_j (\nu)$ denote the partial effect of a small change in the distribution of covariates from $F_{\bs X}$ to $G_{\bs X}$ on the functional $v(F_{\bs Y})$ and let $\boldsymbol{\alpha} (\nu) = (\boldsymbol{\alpha}_1 (\nu), \dots, \boldsymbol{\alpha}_k (\nu))$ be the matrix collecting all $j$ entries. Under the assumption that the conditional distribution of $\bs Y$ given $\bs X$, $F_{\bs Y \mid{\bs X}}$, is unaffected by changes in the law of $\bs X$, the matrix of partial effects of small location shifts in the distribution of $\bs X$ is given by (see Corollary 1 in \citealt{firpo2009unconditional}):}
 %More formally, under the assumption that the conditional distribution of $\bs Y$ given $\bs X$, $F_{\bs Y \mid{\bs X}}$, is unaffected by changes in the law of $\bs X$, the partial effects of small location shifts in the distribution of $\bs X$ on $\nu (F_Y)$ can be written as: %using the vector of average derivatives:
 \begin{equation}\label{eq:upe}
\boldsymbol{\alpha} (\nu) = \int \frac{d \mathbb{E} [RIF(\bs Y; \nu) \mid {\bs X} = \bs {x}]}{d \bs x} dF_{\bs X} (\bs x),
\end{equation}
where $d \mathbb{E} [RIF(\bs Y; \nu) \mid {\bs X} = \bs {x}] / d \bs x$ is understood to indicate the Jacobian matrix of all its first order partial derivatives with respect to $[ \bs x_j]_{j=1}^k$. \cite{firpo2009unconditional} call the quantity $\boldsymbol{\alpha} (\nu)$ in \eqref{eq:upe} as the Unconditional Partial Effect (UPE). \red{By analogy with standard conditional regression coefficients, $\boldsymbol{\alpha} (\nu)$ corresponds to the effect of a small increase in the location of the distribution of the explanatory variables %partial effect 
 on the functional $\nu(F_{\bs Y})$, holding everything else constant.} \redd{It is worth noting that our approach requires $F_{\bs Y \mid \bs X}$ to be \qmo structural\qmcsp in the sense of \qmo invariant to a class of modifications\qmcsp (\citealt{heckman2007econometric}, p. 4848 and Section 4.8) and rules out the presence of potential sources of endogeneity like selection bias in inference.} 
  %In particular, as the conditional expectation of the RIF in \eqref{eq:v0} is a first-order approximation of the effect of $\bs X$ on $\nu(F_{\bs Y})$, the UPE is a local effect estimate of a small change in $\bs X$. The degree of approximation will depend on the empirical context and on the statistic $\nu(F_{\bs Y})$ considered (see \citealt{firpo2009unconditional} and \citealt{Rios2020}).
%$[ \partial \mathbb{E} [RIF(\bs Y_i; \nu) \mid {\bs X} = \bs {x}] / \partial x_j ]$
%holding everything else constant.
%$(\frac{\partial \mathbb{E} [RIF(Y; \theta_\tau) \mid {\bs X} = \bs x]}{\partial x_1}, ..., \frac{\partial \mathbb{E} [RIF(Y; \theta_\tau) \mid {\bs X} = \bs x]}{\partial x_k})$.
%let $X_j$ be a continuous covariate in the $k$-dimensional vector $\mathbf{X}$ and let $\alpha_j (\nu)$ denote the partial effect of a small change in the distribution of covariates from $F_\mathbf{X}$ to $G_\mathbf{X}$ on the functional $\nu(F_Y)$. Collecting all $j$ entries, we construct the k x 1 vector. We can write the unconditional partial effect a(v) as an average derive:
%can be obtained using the conditional expectation conditional expectation of $RIF(y; \nu)$ can 
%Eq. \eqref{eq:upe} states that the effect on $\nu(F_Y)$ of a small change in covariates $\bs X$ is equal to the average derivative of the RIF with respect to the covariates. %It is worth noticing that this approach is concerned solely with the effects of exogenous regressors ensuring that the conditional distribution $F_{Y \mid \bs X}(\cdot)$ is unaffected by small manipulation of the distribution of $\bs X$.
\\
 In the case of a dummy variable $X \in \{0,1\}$, the UPE represents the effect of a small increase in the probability that $X = 1$, namely:
\begin{equation}
\boldsymbol{\alpha} (\nu) = \mathbb{E} [RIF(\bs Y; \nu) \mid X = 1] - \mathbb{E} [RIF(\bs Y; \nu) \mid X = 0].
\end{equation}
%be a univariate response variable with absolutely continuous distribution $F_Y (y)$ and let $\bs X$ denote the covariates vector having support $\mathcal{X} \subseteq \mathbb{R}^k$ and distribution function $F_{\bs X} (\bs x)$. 
%A small location shift in the distribution of $\bs {X}$ can be represented in terms of the counterfactual distribution $G_{\bs X} (\bs {x})$.\\

\redd{In addition, as discussed by \cite{firpo2009unconditional} and \cite{firpo2018decomposing}, the approach based on the RIF can also be used to assess the effect of more general changes in the distribution of covariates on the functional $\nu (F_{\bs Y})$. For example, if the goal of the analysis is to assess how such modifications affect the dependence of $\bs Y$, $\nu$ can be a measure of multivariate scatter or a well-known correlation coefficient, provided that the corresponding influence function can be derived.}

 %As stated in \cite{firpo2009unconditional}, the approach based on the RIF is applicable to a wide range of distributional statistics to capture how they are affected by changes in the distribution of $\bs X$. 
 In the following section we exploit these properties for the analysis of unconditional quantile, M-quantile and expectile regressions associated to multivariate response variables.

\section{Methodology}\label{sec:ummq}
%\section{Definition and Methodology of Multivariate Unconditional M-quantile Regression}\label{sub:met}
In this section, we generalize the univariate approach of \cite{firpo2009unconditional} %to a multivariate context 
 by developing a unifying regression method to model unconditional quantiles, M-quantiles and expectiles of vector-valued responses \red{for exploring the effects of covariates} using the RIF approach illustrated in Section \ref{sec:pre}. Firstly, we introduce the Huber's multidimensional M-function for estimating multivariate quantiles, M-quantiles and expectiles by varying the value of the tuning constant in an appropriate way. %Firstly, we recall the definition of multivariate M-quantile proposed by \cite{kokic2002new} and 
 Then, in order to assess the effect of the covariates at different parts of the unconditional distribution of the dependent variable, we introduce the Unconditional Quantile, M-Quantile and Expectile Partial Effect (UQPE, UMQPE, UEPE) %, depending on the values of the tuning constant in the multidimensional Huber function,
 and establish the asymptotic properties of the corresponding estimators.
 The Huber's function, in the univariate case, has been used by \cite{breckling1988m} to define the M-quantile, extending the concept of M-estimation of \cite{huber1964} to the quantile framework. Huber M-quantiles are a generalized form of M-estimators that include in a single modeling approach, quantiles and expectiles through a tuning constant to adjust the robustness of the estimator in the presence of outliers. In higher dimensions, extending these univariate notions to multivariate data is not a trivial task since there does not exist a natural ordering in $p$ dimensions, $p > 1$. Already in their proposal, \cite{breckling1988m} considered the multivariate extension of Huber's function to estimate M-quantiles by
 %has been \cite{breckling1988m} which 
 %originally considered the multidimensional Huber function as the loss function of a minimization problem that leads to the definition of multivariate M-quantiles. However, extending the univariate notion of M-quantiles to multivariate data is not a trivial task because there does not exist a natural ordering in $p$ dimensions, $p > 1$. To overcome this issue, the authors 
 considering %addressed this problem by introducing 
 a directional unit norm vector %in the Euclidean $p$-dimensional space 
 to set up a suitable ordering procedure for multidimensional data. Further, \cite{kokic2002new} generalized their approach by introducing a weighted estimating equation based on the multidimensional Huber's influence function that encompasses multivariate quantiles, M-quantiles and expectiles, depending on the value of the related tuning constant. \red{Their method aimed to provide a robust and simple to implement technique of summarising and conveying valuable information about multidimensional data.} %and hence they have applications in outlier detection and for performance measurement condense Their interest aimed
 Here and in what follows, in order to present a unified unconditional regression approach to model multivariate quantiles, M-quantiles and expectiles, we adopt the approach of \cite{kokic2002new} based on the multidimensional Huber's function. 

Formally, the multidimensional Huber's influence function in \cite{hampel2011robust} is defined as:
\begin{equation}\label{eq:mhub}
     \Psi(\bs{r}) = \left\{ \begin{array}{ll}
      \frac{\bs{r}}{c}, & \Vert \bs{r} \Vert < c \\
      \frac{\bs{r}}{\Vert \bs{r} \Vert}, & \Vert \bs{r} \Vert \geq c \\
\end{array} \right., \qquad \bs{r} \in \mathbb{R}^p, \bs{r} \neq \bs 0
\end{equation} 
%(\citealt{serfling2008mahalanobis})
\red{with additionally $\Psi(\bs 0) = \bs 0$ and} where $c \geq 0$ is the tuning constant that can be adjusted to trade robustness for efficiency, with increasing robustness when it is chosen to be close to 0 and increasing efficiency when it is chosen to be large. %As one can see, the function in \eqref{eq:mhub} tries to cut down the influence of possibly abnormal observations by transforming each point outside the hypersphere of radius $c$ to its nearest point on it and leaving those inside unaffected.
 To show how one can use
%To show how \cite{kokic2002new} considered 
 the $\Psi(\bs{r})$ function in \eqref{eq:mhub} to estimate multivariate quantiles, M-quantiles and expectiles, we introduce the following additional notation. Let $\mathcal{Y} = \mathbb{R}^p$, consider a continuous $p$-dimensional random variable ${\bs Y}$ %with absolutely continuous distribution function $F_{\bs Y}$ %and let $\bs y$ denote its realization. 
 and let $\bs u$ denote a unit norm direction vector ranging over the $p$-dimensional unit sphere $\mathcal{S}^{p-1} = \{ \bs{z} \in \mathbb{R}^p : || \bs{z} || = 1 \}$, where $||\cdot||$ denotes the Euclidean norm. Following \cite{kokic2002new}, for a general value of $c$, we obtain the $\tau$-th multivariate M-quantile of $\bs Y$ in the direction of $\bs u$, $\bold \theta_{\tau, \bs{u}}$, with $\tau \in (0, \frac{1}{2}]$, by satisfying the equation:
\begin{equation}\label{eq:min}
\int \eta_{\delta} (\varphi) \Psi(\bs{y} - \bold \theta_{\tau, \bs{u}}) d F_{\bs Y} (\bs y) = {\bs 0},
\end{equation}
where 
\[ \eta_{\delta} (\varphi) =  \left\{
\begin{array}{ll}
      (1 - \cos \varphi)^\delta \zeta + 2\tau, & \varphi \in (-\frac{\pi}{2}, \frac{\pi}{2}) \\
      -(1 - \cos \varphi)^\delta \zeta + 2(1-\tau), & \varphi \in [-\pi, -\frac{\pi}{2}] \cup [\frac{\pi}{2}, \pi], \\
\end{array} 
\right. \]
is a weighting function with $\zeta = 1 - 2\tau$, $\delta > 0$ and $\varphi$ being the angle between $\bs{Y} - \bold \theta$ and $\bs{u}$, so $\cos \varphi = \frac{(\bs{Y} - \bold \theta)' \bs{u}}{\Vert \bs{Y} - \bold \theta \Vert}$. \red{Evidently, one can also consider non-normalized directions by explicitly writing $\cos \varphi = \frac{(\bs{Y} - \bold \theta)' \bs{u}}{\Vert \bs{Y} - \bold \theta \Vert \Vert \bs{u} \Vert}$.}

 The function $\eta_{\delta} (\varphi)$ gives asymmetric weights to the residual $\bs Y - \bold \theta$ depending both on its length and the angle it forms with $\bs u$. 
 Most importantly, the tuning constant $c$ %The tuning constant $c \geq 0$ in \eqref{eq:mhub} can be used to trade robustness for efficiency with increasing robustness when it is chosen to be close to 0 and increasing efficiency when it is chosen to be large. Specifically, it 
 determines where the weighted scheme based on $\eta_{\delta} (\varphi)$ defines multivariate quantiles when $c = 0$ and yields multivariate expectiles as $c \rightarrow \infty$, which could be particularly fruitful when the use of outlier-robust estimation methods is not justified but there is still interest in modeling the entire distribution of $\bs Y$ (\citealt{tzavidis2010m}). \red{Hence, multivariate expectiles inherit the efficiency properties of standard univariate expectiles and, at the same time, consider the dependence structure between the components of the variable analyzed}. Computationally, an estimate of $\bold \theta_{\tau, \bs{u}}$ in \eqref{eq:min} can be efficiently obtained by using Iteratively Reweighted Least Squares (IRLS, \citealt{breckling2001note}). %\textcolor{red}{The definition of multivariate M-quantile in \eqref{eq:min} overcomes the shortcoming of the approach of \cite{breckling1988m} which was causing the estimated M-quantiles to be situated outside the convex hull of the data in certain situations.} 
 Clearly, the multivariate M-quantile in \eqref{eq:min} includes the traditional notion of univariate M-quantile. \redd{Indeed, if $p=1$, $u = 1$ and $\delta = 1$, then \eqref{eq:min} reduces to the estimating equation of the univariate M-quantile, $\theta_\tau$, because $\cos \varphi = \textnormal{sgn}(Y - \theta_\tau)$, which implies that $\eta_{\delta} (\varphi) = 1 - \zeta \textnormal{sgn}(Y - \theta_\tau)$.}
%Indeed, if $p=1$ and $u = 1$, then \eqref{eq:min} reduces to the estimating equation of the univariate M-quantile, $\theta_\tau$, because %\eqref{eq:mhub} coincides with the well-known Huber influence function (\citealt{huber1964}) and 
 %$\cos \varphi = \textnormal{sgn}(Y - \theta_\tau)$, which implies that $\eta_{\delta} (\varphi) = 1 - \zeta \textnormal{sgn}(Y - \theta_\tau)$. %Again, in the univariate case, \eqref{eq:mhub} includes the well-known Huber influence function (\citealt{huber1964}) $\psi(r) = u \mathbf{1}_{(\mid r \mid \leq c)} + c \, \textnormal{sign} (r) \mathbf{1}_{(\mid r \mid > c)}$.}
 \red{Throughout the rest of the paper we set $\delta = 1$, but other values of $\delta$ are possible (see \citealt{kokic2002new} and the Supplementary Materials for this article).}
 %unless specified otherwise, or various values but other choices are possible
  %we will always use directions with Euclidean norm 1
  
   %For a given $\bs{u}$ and $\delta$, the space is divided in two regions by a hyperplane orthogonal to $\bs{u}$ and passing through $\bold \theta_{\tau, \bs{u}}$, which allows for a direct interpretation of $\bs u$. The function $\eta_{\delta} (\alpha)$ gives asymmetric weights to the residual $\bs Y - \bold \theta$ depending both on its length and the angle it form with $\bs u$. This definition has several desirable properties. 
\red{In comparison with other settings in the literature, the considered approach remedies the shortcoming of the \cite{chaudhuri1996geometric} and \cite{herrmann2018multivariate} definitions of geometric quantile and expectile which have recently been criticized because they can be situated outside the support of $\bs Y$ for extreme levels of $\tau$ (see \citealt{girard2017intriguing, konen2021multivariate}). Moreover, \cite{breckling1988m} M-quantiles, which can be shown to coincide with the definition in \cite{chaudhuri1996geometric} as a particular case, are also subject to the same problem. The definition in \eqref{eq:min}, on the contrary, lead to multivariate M-quantiles always situated within the convex hull of the sample data (for a comparative analysis between the considered definition of multivariate (M-)quantile and the geometric quantile, see the Supplementary Materials).}
%For extreme indices $\tau$
%Secondly, this de:nition would again lead to quantiles situated outside the convex hull of the data,
%Utilizing a framework comparable to the one introduced in Chaudhuri (1996) to generalize quantiles
%In the presented examples, we utilized these simulation-based approximations to contrast geometric expectiles to the geometric quantiles introduced in Chaudhuri (1996) as well as univariate expectiles and quantiles
%In [3] a definition of multivariate geometric quantiles is presented which, however, can be shown to coincide with a special case of the definition in [1] and hence is subject to the same problem
%Chaudhuri (1996) presents a de:nition of multivariate geometric quantiles which, however, can be shown to coincide with a special case of the de:nition of Breckling and Chambers (1988) and hence is subject to the same problem.
 
 \red{A particular issue in this context may be the choice of the direction $\bs u$, which is often selected on the basis of the empirical problem at hand to produce meaningful results (see \citealt{paindaveine2011directional, kong2012quantile} and \citealt{merlo2022marginal}). For example, \cite{fraiman2012quantiles} and \cite{torres2017directional} consider a set of directions using the principal components obtained from a PCA to build a useful tool for exploratory data analysis and visualize important features of multidimensional data. 
 %, in order to identify outliers in multivariate data,  PCA directions
  In classification analysis, \cite{farcomeni2020directional} focus on a single optimal direction %of analysis and determine the optimal one 
 that minimizes the misclassification error meanwhile, \cite{geraci2019quantile} define the ``allometric direction'' for risk classification of abnormal multivariate anthropometric measurements. %In classification analysis, \cite{farcomeni2020directional} focus on a single direction of analysis and determine the optimal one that minimizes the misclassification error. 
 %From a practical perspective, different contexts or phenomena could lead to consider different particular directions of interest. For instance, in financial portfolio theory, the direction given by the portfolio weights is of particular interest because it takes into account the dependence among the components of the portfolio 
 %in financial portfolio theory where the portfolio return function is given by a linear combination of several assets' returns, %of several marginal losses, 
 %the direction given by the portfolio weights (i.e. the percentages of capital invested in each asset within the portfolio) %chosen by the investor 
 %is of particular interest because it takes into account the dependence among the components of the portfolio 
 %(\citealt{torres2015directional}). %set the direction $\bs u$ to be equal to the portfolio weights vector chosen by the investor, accounting for the  
 %On the other hand, in environmental modeling, \cite{torres2017directional} propose to use the first PCA direction for identification and better visualization of extreme environmental events, such as floods, storms, and droughts.
 }

   %One option is also to use the principal component of a Principal Component Analysis by maximizing the variance of the projected data $\bs u' \bs Y$ as proposed in \cite{korhonen1998ordinal} and \cite{geraci2019quantile}, or to consider the direction by minimizing a measure of skewness.

 %following the works of \cite{paindaveine2011directional, kong2012quantile, torres2015directional, geraci2019quantile, farcomeni2020directional} and \cite{merlo2022marginal}, the choice of the most appropriate directions for the analysis depends on the problem under consideration and the research question of interest. % of the data analysis. 
%the choice of $\bs u$ should be data driven and shall take into account the research question at hand.\\
%From an applied perspective, choosing an adequate number of states is a crucial aspect of the data analysis which  data structure and 

%\textcolor{red}{Moreover, when theoretically all directions are investigated simultaneously, the proposed multivariate M-quantiles generate centrality regions and contours which allow us to assess the location, spread and shape of the entire distribution of the responses. In particular, M-quantile contours adapt to the shape of the distribution of interest and summarize the information carried by each direction-dependent M-quantile, describing the dependence between the responses and specific features of multivariate data. 
 
 %When theoretically all directions are investigated simultaneously,
 When the directional approach is adopted, considering theoretically all directions in $\mathcal{S}^{p-1}$ simultaneously yields multivariate M-quantiles centrality regions, which allow us to provide a visual description of the location, spread, shape and dependence between the responses distribution. These quantities are of crucial interest as they are able to adapt to the underlying shape of the distribution of $\bs Y$ without being constrained to particular shapes, such as convex bodies or ellipses (\citealt{breckling2001note}). \red{Specifically, for a given level $\tau \in (0, \frac{1}{2}]$, by moving the direction $\bs u$ around the whole $\mathcal{S}^{p-1}$, the resulting set of corresponding multivariate M-quantiles generates the $\tau$-th M-quantile region embedded within the $p$-dimensional Euclidean space, $R_\tau \subset \mathbb{R}^{p}$, defined as the collection whose vertices are:}
%a closed region embedded within the $p$-dimensional Euclidean space. Thus, the $\tau$-th M-quantile region, $R_\tau \subset \mathbb{R}^{p}$, is defined as the set whose vertices are:%Formally, if the distribution of $\bs Y$ is absolutely continuous, we may restrict to $\tau \in (0, \frac{1}{2}]$ and define the $\tau$-th M-quantile region, $R_\tau \subset \mathbb{R}^{p}$, as the set whose vertices are:
\begin{equation}\label{eq:mqcon}
R_\tau = \{ \bold \theta_{\tau, \bs{u}} \, \big| \, {\bs u} \in \mathcal{S}^{p-1} \}.
\end{equation}
The region in \eqref{eq:mqcon} is a closed surface and the corresponding M-quantile contour of order $\tau$ is defined as the boundary $\partial R_\tau$ of $R_\tau$. \red{The position, curvature, spread and orientation of such objects all reflect important characteristics of the data that are relevant to researchers for exploratory data analysis.} %as well as a statistical diagnostic tool to evaluate the outlier sensitivity 
%\textcolor{red}{For any value of $\tau$ and $c$, the definition of multivariate M-quantile in \eqref{eq:min} ensures that the contours are always located within the convex hull of the data, thus overcoming the drawback of the approaches of \cite{breckling1988m, chaudhuri1996geometric} and \cite{herrmann2018multivariate}.}
%Such quantities are of crucial interest as they are able to detect covariate-dependent features of the distribution of the responses given $\bs X$, %Such quantities are of crucial interest because they completely characterize the distribution of $\bs Y$ conditional on $\bs X$. 
 %Specifically, they characterize the effect of the regressors on the location, spread, skewness and kurtosis of the responses distribution, 
 %while ensuring robustness to outlying data. 
 \red{M-quantile curves can thus be used for exploratory data analysis as well as a statistical diagnostic tool to evaluate the results to different values of the tuning constant.} For fixed $\tau$, %when $c \rightarrow 0$, M-quantile contours reduce to directional quantile envelopes illustrated in \cite{kong2012quantile}; on the other hand, 
 when $c=0$ it defines quantile contours and it generates expectile contours when $c \rightarrow \infty$. Meanwhile, for any $c \geq 0$, the contours are nested as $\tau$ increases. As $\tau \rightarrow 0$, instead, the $\tau$-th M-quantile contour approaches the convex hull of the sample data providing valuable information about the extent of extremality of points (see \citealt{serfling2002quantile} and \citealt{kokic2002new}).%This is a fundamental property as it allows a probability based ordering of the data (see \citealt{serfling2002quantile}) and enables us to identify points that lie outside the estimated contour as having jointly abnormal outcomes, conditional on $\bs X$.

To build our model, it follows from \eqref{eq:min} and \cite{hampel2011robust} that the IF for the unconditional multivariate M-quantile $\bold \theta_{\tau, \bs{u}}$ is defined as:
\begin{equation}\label{eq:mIF}
IF({\bs y}; \bold \theta_{\tau, \bs{u}}) = {\bs M} {(\bold \theta_{\tau, \bs{u}})}^{-1} \eta_{\delta} (\varphi) \Psi ({\bs y} - \bold \theta_{\tau, \bs{u}})
\end{equation}
with ${\bs M (\bold \theta_{\tau, \bs{u}})}$ being the $p \times p$ matrix given by:
\begin{align}\label{eq:M}
{\bs M} (\bold \theta_{\tau, \bs{u}}) = - \int \nabla_{\bold \theta_{\tau, \bs{u}}} \big( \eta_{\delta} (\varphi)  \Psi(\bs{y} - \bold \theta) \big) dF_{\bs Y} (\bs{y}),
\end{align}
%&= - \int \big( \Psi ({\bs y} - \bold \theta_{\tau, \bs{u}}) \nabla_{\bold \theta_{\tau, \bs{u}}} \eta_{\delta} (\alpha) + \eta_{\delta} (\alpha) \nabla_{\bold \theta_{\tau, \bs{u}}} \Psi(\bs{y} - \bold \theta) \big) dF_{\bs Y} (\bs{y})
%\begin{equation}\label{eq:M}
%{\bs M} (\bold \theta_{\tau, \bs{u}}) = \int \eta_{\delta} (\alpha) \nabla_{\bold \theta_{\tau, \bs{u}}} \Psi(\bs{y} - \bold \theta) dF_{\bs Y} (\bs{y}),
%\end{equation}
where $\nabla_{\bold \theta_{\tau, \bs{u}}} \big( \cdot \big)$ is the $p \times p$ matrix of first order derivatives of $ \eta_{\delta} (\varphi)  \Psi(\bs{y} - \bold \theta)$ in \eqref{eq:min} with respect to $\bold \theta$ evaluated at $\bold \theta_{\tau, \bs{u}}$. Then, following the idea in \eqref{eq:RIF0}, the RIF is obtained from \eqref{eq:mIF} by adding back the multivariate M-quantile $\bold \theta_{\tau, \bs{u}}$:
\begin{equation}\label{eq:mRIF}
RIF({\bs y}; \bold \theta_{\tau, \bs{u}}) = \bold \theta_{\tau, \bs{u}} + IF({\bs y}; \bold \theta_{\tau, \bs{u}}).
\end{equation}
Two remarks are worth noticing. Firstly, \eqref{eq:mRIF} generalizes the RIF of the multivariate quantile when $c = 0$, where the matrix of first order derivatives of $\Psi (\bs{r})$ is equal to:
\begin{equation}
\frac{d \Psi(\bs{r})}{d \bs r} = \frac{1}{\Vert \bs{r} \Vert} \Big\{ \mathbf{I}_p - \frac{\bs{r} \bs{r}'}{\Vert \bs{r} \Vert^{2}} \Big\},
\end{equation}
with $\mathbf{I}_p$ being the identity matrix of dimension $p$, and it coincides with the RIF of the multivariate expectile when $c \rightarrow \infty$, which implies $\frac{d \Psi(\bs{r})}{d \bs r} \propto \mathbf{I}_p$. %In this case, one can immediately recognize that \eqref{eq:mRIF} is the RIF for the mean vector if $\tau = \frac{1}{2}$.
 %$RIF(\bs{y}; T, F_\theta) = \bs{y}$ as particular cases. 
 Secondly, when $p=1$ and $u=1$, \eqref{eq:mRIF} reduces to the univariate RIF of standard M-quantiles which, in turn, includes the RIF of the quantile in \cite{firpo2009unconditional} and the RIF of the expectile for $c$ arbitrarily large. %Indeed, when $\tau = \frac{1}{2}$ and $c \rightarrow \infty$, we obtain $\eta_{\delta} (\alpha) = 1$ and \eqref{eq:mRIF} reduces to $RIF(\bs{y}; \bold \theta_{\frac{1}{2}, \bs{u}}) = \bs{y}, \, \,\forall \, \bs u \in \mathcal{S}^{p-1}$, which is the RIF of the mean vector. %evaluated at $\bs y$.\\
 Further, using this approach we are able to investigate the correlation structure of multivariate responses at different values of $\tau$. More in detail, to study the association between multiple outcomes we analyze the covariance matrix of the RIF in \eqref{eq:mRIF} which, by simple calculations, can be written as:
\begin{equation}\label{eq:cor}
\bs \Delta (\bold \theta_{\tau, \bs{u}}) = \mathbb{E} [ IF({\bs Y}; \bold \theta_{\tau, \bs{u}}) IF({\bs Y}; \bold \theta_{\tau, \bs{u}})'].
\end{equation}
Given $\bs u$, $\tau$ and $c$, the off-diagonal elements of $\bs \Delta (\bold \theta_{\tau, \bs{u}})$ provide a measure of tail correlation between the components of $\bs Y$.
 
In a regression framework where covariates $\bs X$ are available, from \eqref{eq:mRIF} we define the unified unconditional regression model as follows:
%\begin{equation}\label{eq:RIFmm}
%\mathbb{E} [RIF(\bs Y; \bold \theta_{\tau, \bs{u}}) \mid {\bs X} = \bs {x}] = \bold \theta_{\tau, \bs{u}} + \mathbb{E} [IF({\bs Y}; \bold \theta_{\tau, \bs{u}}) \mid {\bs X} = \bs {x}].
%\end{equation}
\red{
\begin{equation}\label{eq:RIFmm}
\mathbb{E} [RIF(\bs Y; \bold \theta_{\tau, \bs{u}}) \mid {\bs X} = \bs {x}] = m_{\bold \theta_{\tau, \bs{u}}} (\bs x),
\end{equation}
where $m(\cdot)$ is an unknown function of explanatory variables $\bs X$ to be estimated.
}
Our objective is to identify how changes in the distribution of $\bs X$ affect the multivariate quantile, M-quantile and expectile of the unconditional distribution of $\bs Y$. Following \eqref{eq:upe}, for a given level $\tau$, direction $\bs u$, and constant $c \geq 0$, the Unconditional M-Quantile Partial Effect (UMQPE), $\boldsymbol{\alpha}_{\tau, \bs{u}}$, is formally defined as:
%\begin{split}
\begin{equation}\label{eq:mpart}
\boldsymbol{\alpha}_{\tau, \bs{u}} = \int \frac{d \mathbb{E} [RIF({\bs Y}; \bold \theta_{\tau, \bs{u}}) \mid {\bs X = \bs x}]}{d \bs x} dF_{\bs X} (\bs x).
\end{equation}
%\end{split}
%&= {\bs M} {(\bold \theta_{\tau, \bs{u}})}^{-1} \int \frac{d \mathbb{E} [\eta_{\delta} (\varphi) \Psi ({\bs Y} - \bold \theta_{\tau, \bs{u}}) \mid {\bs X = \bs x}]}{d \bs x} dF_{\bs X} (\bs x).
%\textcolor{red}{Moreover, we define the  when $c = 0$ and $c \rightarrow \infty$, respectively.}
It is worth noting that the proposed approach has several appealing properties. Firstly, the UMQPE in \eqref{eq:mpart} is easy to compute as it does not depend on the density of $\bs Y$ which would entail the use of nonparametric density estimation procedures (\citealt{kokic2002new}). Secondly, this methodology allows us to directly control the robustness to outliers and estimation efficiency by means of the tuning constant $c$, %, with increasing robustness when $c$ is chosen to be close to $0$ and increasing efficiency for $c$ large. 
%In particular, 
 \red{i.e.,} when $c = 0$ we have the UQPE and when $c \rightarrow \infty$ we have the UEPE. \red{In practice, the UMQPE indicates the effect of increasing the years of schooling or income, say, across the consumption and wealth distributions, as we will discuss in the next section.}

Before concluding this section, it is relevant to compare the proposed unconditional regression model with the standard conditional regression approach. % as defined by \cite{breckling1988m} and \cite{kokic2002new}.
 Suppose that $\bs Y = h(\bs X) + \bs \epsilon$, where $h: \mathcal{X} \rightarrow \mathbb{R}^p$ is an unknown function with bounded first partial derivatives and $\bs \epsilon$ is a $p$-dimensional random error independent of $\bs X$. For a given $\tau$, $\bs u$ and $c \geq 0$, we denote the $\tau$-th multivariate M-quantile of $\bs Y$ conditional on $\bs X = \bs x$, $\bold \theta_{\tau, \bs{u}} (\bs x) = h(\bs x)$, as the solution of the following estimating equation:
\begin{equation}\label{eq:minx}
\int \eta_{\delta} (\varphi) \Psi(\bs{y} - \bold \theta_{\tau, \bs{u}} (\bs x)) d F_{\bs Y \mid \bs X} (\bs y \mid \bs x) = {\bs 0}.
\end{equation}
Consequently, the effect of a small change in $\bs X$ on the conditional M-quantile of $\bs Y$, $\bold \theta_{\tau, \bs{u}} (\bs x)$, which we denote as Conditional M-Quantile Partial Effect (CMQPE), is given by:
\begin{equation}\label{eq:cmmqpe}
\boldsymbol{\alpha}_{\tau, \bs{u}} (\bs x) = \frac{d h(\bs x)}{d \bs x}.
\end{equation}
In order to clarify the interpretation of \eqref{eq:mpart}, following \cite{firpo2009unconditional}, we provide a useful representation of the UMQPE in terms of the conditional distribution of $\bs Y$ given $\bs X$ and show how it is related to the CMQPE in \eqref{eq:cmmqpe}.
 Let $\bs W_{\tau, \bs{u}} : \mathcal{X} \rightarrow \mathbb{R}^{p \times p}$ define the weighting matrix function:
\begin{equation}\label{eq:w}
\bs W_{\tau, \bs{u}} (\bs x) = {\bs M} {(\bold \theta_{\tau, \bs{u}})}^{-1} \mathbb{E}[\nabla \big( \eta_{\delta} (\varphi) \Psi (\bs Y - \bold \theta_{\tau, \bs{u}}) \big) \mid \bs X = \bs x]
\end{equation}
%\nabla_{\bold \theta_{\tau, \bs{u}}}
and let $s_{\tau, \bs u}$ be \red{an auxiliary function $s_{\tau, \bs u} : \mathcal{X} \rightarrow (0,1)$ required to establish the link between the UMQPE and the CMQPE. The mapping $s_{\tau, \bs u}$ can be thought as a ``matching'' function indicating where the unconditional multivariate M-quantile $\bold \theta_{\tau, \bs{u}}$ falls in the conditional distribution of $\bs Y$ given covariates, i.e.:}
%that maps each conditional multivariate M-quantile onto the unconditional multivariate M-quantile of $\bs Y$, i.e.:
\begin{equation}\label{eq:s}
s_{\tau, \bs u} (\bs x) = \{ \tilde \tau : \bold \theta_{\tilde \tau, \bs{u}} (\bs x) = \bold \theta_{\tau, \bs{u}} \}.
\end{equation}
\redd{In this work, we require that, under the condition that $F_{\bs Y \mid \bs X} (\bs y \mid \bs x)$ is continuous and strictly monotonic (\citealt{breckling1988m}), $s_{\tau, \bs u} (\bs x)$ is a singleton, that is, there is only one $\tilde \tau$ that satisfies the equality in \eqref{eq:s}. Generally, three situations may occur: $s_{\tau, \bs u} (\bs x)$ is unique, it is interval-valued or it is empty for some $\bs x$ (\citealt{alejo2021conditional}). From a practical point of view, for fixed $\bs x$, if $s_{\tau, \bs u} (\bs x)$ is interval-valued, then a possible solution would be to take the average of all $\tilde \tau$'s inside that particular interval. If, on the other hand, $s_{\tau, \bs u} (\bs x)$ is empty, one could impute $\tilde \tau$ by taking the average of the nearest neighbors to the covariate value $\bs x$. Refer to the Supplementary Materials for a discussion and an estimator of the matching function.}
%yields the $\tilde \tau$-th M-quantile level such that the unconditional and conditional M-quantile of $\bs Y$ matches.

 The next Theorem establishes a link between the UMQPE and the CMQPE. %key result establishing a relationship between the UMMQPE and the CMMQPE.
\begin{theorem}\label{thm:wavem}
Assume that $\bs Y = h(\bs X) + \bs \epsilon$ where $h(\cdot)$ is an unknown function with bounded first partial derivatives and $\bs \epsilon$ is an error term independent of $\bs X$. %$\bs X$ is independent of $\bs \epsilon$ and $h$ is monotonic in $\bs \epsilon$. 
 For a given $\tau \in (0, \frac{1}{2}]$, $\bs{u} \in \mathcal{S}^{p-1}$ and $c \geq 0$, the UMQPE $\boldsymbol{\alpha}_{\tau, \bs{u}}$ can be written as:
%\begin{equation}\label{eq:wavem1}
%UMMQPE_{\tau, \bs{u}} = \mathbb{E}[\bs W_{\tau, \bs{u}} (\bs X) \, \frac{\partial h(\bs X)}{\partial \bs x}],
%\end{equation} 
%or equivalently: %, it can be defined as: % a weighted average of the $CMMQPE_{\tau, \bs u} (\bs x)$:
\begin{equation}\label{eq:wavem2}
\boldsymbol{\alpha}_{\tau, \bs{u}} = \mathbb{E}[\bs W_{\tau, \bs{u}} (\bs X) \, \boldsymbol{\alpha}_{s_{\tau, \bs u} (\bs X), \bs u} (\bs X)],
\end{equation}
where the expectation is taken over the distribution of $\bs X$.
\end{theorem}
\begin{proof}
See Proof of Theorem \ref{thm:wavem} in the Supplementary Materials.
\end{proof}
From Theorem \ref{thm:wavem} there follow several interesting considerations. Firstly, it formally shows that, unlike conditional means which average up to the unconditional mean thanks to the law of iterated expectations, conditional multivariate M-quantiles do not average up to their unconditional counterparts. On the contrary, the UMQPE is equal to a weighted average, over the distribution of the covariates, of the CMQPE at the $s_{\tau, \bs{u}} (\bs X)$-th conditional M-quantile level corresponding to the $\tau$-th unconditional M-quantile of $\bs Y$ in the direction $\bs u$. Secondly, it generalizes the result in \cite{firpo2009unconditional} to the multivariate setting which also holds in the univariate case when $p=1$ and $u = 1$. %, and, because Theorem \ref{thm:wavem} holds for any $c \geq 0$, both the UQPE and UEPE can be expressed in terms of the conditional quantile and expectile partial effects as in \eqref{eq:wavem2}.  
 Furthermore, Theorem \ref{thm:wavem} is useful for interpreting the parameters of the proposed regression method. For instance, in a linear model where $h(\bs X) = \boldsymbol{\beta}' \bs X$, the UMQPE and CMQPE are both equal to the matrix of regression coefficients $\boldsymbol{\beta}$ for any choice of $\tau$ and $\bs u$. More generally, the UMQPE and CMQPE will be different depending on the structural form of $h(\cdot)$ and the distribution of $\bs X$ as described in Theorem \ref{thm:wavem}. %(see also \cite{firpo2009unconditional} for other examples of the structural form of $h$).}
%On the contrary, the UMMQPE will depen
 %\red{Lastly, it may be noted that to exploit Theorem \ref{thm:wavem} in more general cases, such as heteroskedasticity, a practical solution would be to consider a flexible form of $h(\cdot)$ and estimate the conditional expectation of the RIF using nonparametric regression approaches (see \citealt{newey1994asymptotic}).}
 \redd{Lastly, in the presence of heteroskedastic errors, one can consider the linear heteroskedastic model $\bs Y = \boldsymbol{\beta}' \bs X +\mbox{diag}(\boldsymbol{\gamma}' \bs X) \epsilon$ which is often used in economics, with $\boldsymbol{\gamma}$ being a matrix of coefficients. In this case, to exploit Theorem \ref{thm:wavem} a practical solution would be to adopt a flexible form of $h(\cdot)$ and estimate the conditional expectation of the RIF using nonparametric regression approaches. Alternatively, a possible way to alleviate heteroskedasticity is to apply a transformation to the original response variable.}%of interest.

\subsection{Estimation}\label{sub:est}
In this section, we discuss the estimation of the UQPE, UMQPE and UEPE using the RIF regression approach. \redd{Following \cite{firpo2009unconditional}, we suggest two methods for modeling the conditional expectation of the RIF. For a given $\tau$, direction $\bs u$ and $c \geq 0$, the first one assumes that} the RIF in \eqref{eq:RIFmm} is linear in the covariates, $m_{\bold \theta_{\tau, \bs{u}}} (\bs X) = \bold \beta' (\bold \theta_{\tau, \bs{u}}) \bs X$ and estimate $\boldsymbol{\alpha}_{\tau, \bs{u}}$ in \eqref{eq:mpart} via an OLS regression of the $RIF({\bs Y}; \bold \theta_{\tau, \bs{u}})$ as a dependent variable onto the covariates $\bs X$ by using a two-step procedure. Specifically, an estimate $\widehat{\bold \theta}_{\tau, \bs{u}}$ of $\bold \theta_{\tau, \bs{u}}$ is obtained by solving \eqref{eq:min} via IRLS, substitute $\widehat{\bold \theta}_{\tau, \bs{u}}$ in \eqref{eq:mRIF} and then estimate $\boldsymbol{\alpha}_{\tau, \bs{u}}$ by regressing the $RIF({\bs Y}; \widehat{\bold \theta}_{\tau, \bs{u}})$ on ${\bs X}$. Let $(\bs Y_i, \bs X_i)$, $i = 1,\dots,n$, denote a random sample of size $n$, the estimator of the UMQPE in \eqref{eq:mpart}, $\widehat{\boldsymbol{\alpha}}_{\tau, \bs{u}}$, is defined as follows:
\begin{equation}\label{eq:ummqpee}
\widehat{\boldsymbol{\alpha}}_{\tau, \bs{u}} = \widehat{\Omega}^{-1}_{\bs X} \frac{1}{n} \sum_{i=1}^n \{ \bs X_i RIF'(\bs Y_i; \widehat{\bold \theta}_{\tau, \bs{u}}) \}.
\end{equation}
where $\widehat{\Omega}_{\bs X} = \frac{1}{n} \sum_{i=1}^n \bs X_i \bs X'_i$. %In addition, exploiting the results in Section \ref{sec:pre} and from \eqref{eq:mRIF}, the estimator of the policy effect in \eqref{eq:pe1} for the multivariate M-quantile $\boldsymbol{\theta}_{\tau, \bs{u}}$, $\widehat{\bold \pi}^{(\ell)}_{\tau, \bs{u}} = \widehat{\bold \pi}^{(\ell)} (\boldsymbol{\theta}_{\tau, \bs{u}})$, is given by
%\begin{equation}\label{eq:mpee}
%\widehat{\bold \pi}^{(\ell)}_{\tau, \bs{u}} = \widehat{\boldsymbol{\alpha}}_{\tau, \bs{u}}' \frac{1}{n} \sum_{i=1}^n \big( \ell (\bs X_i) - \bs X_i \big),
%\end{equation}
%where $\widehat{\boldsymbol{\alpha}}_{\tau, \bs{u}}$ is the estimator of the UMQPE in \eqref{eq:ummqpee}. 
 
 \redd{The second method estimates non-parametrically $\mathbb{E} [RIF(\bs Y; \bold \theta_{\tau, \bs{u}}) \mid {\bs X}]$ using regression B-splines to account for nonlinear effects of $\bs X$ on the RIF. Once we have regressed $RIF({\bs Y}; \widehat{\bold \theta}_{\tau, \bs{u}})$ on the basis functions of the original covariates, %matrix B for the original matrix of independent variables, x
 as the object of interest is the average of $\frac{d \mathbb{E} [RIF({\bs Y}; \bold \theta_{\tau, \bs{u}}) \mid {\bs X = \bs x}]}{d \bs x}$, to obtain the UMQPE we simply take derivative of B-spline basis functions and average them with respect to $\bs X$. Note also that other nonparametric approaches can be adopted such as power series estimators (\citealt{newey1994asymptotic}) or orthogonal polynomials.}
  
 %Similarly, 
 \redd{Finally,} the estimators of the UQPE and UEPE related to \eqref{eq:mpart} %, and those related to the policy effect 
 can be obtained following the same procedure by setting $c = 0$ and $c$ large enough such that $|| \bs Y_i - \bold \theta_{\tau, \bs{u}}  || < c$, $\forall i = 1,\dots,n$, respectively. %Similarly to the UQPE and UEPE, the estimators of the policy effect for the multivariate quantiles and expectiles can be obtained following the same procedure by setting $c = 0$ and $c$ large enough such that $|| \bs Y_i - \bold \theta_{\tau, \bs{u}}  || < c$, $\forall i = 1,\dots,n$, respectively.}
 %From an applied perspective
 \red{If the researcher is interested in modeling multivariate quantiles or expectiles, one can simply set $c = 0$ or $c$ to a relatively large value. In all other cases, the UMQPE estimator in \eqref{eq:ummqpee} requires choosing a reasonable value for the tuning constant, based on the data structure and the aim of the analysis.} %prior knowledge In all other cases, the UMQPE estimator in \eqref{eq:ummqpee} depends on the unknown tuning constant $c$ of the Huber's influence function. 
  
  The choice of an appropriate value for $c$ is not straightforward. Ideally, it should be data-driven and account for possible outliers in the data. In the literature on univariate M-estimation, $c$ can be either fixed a-priori or defined by the data analyst to achieve a specified asymptotic efficiency under normality (\citealt{huber2009robust}), maximize the asymptotic efficiency (\citealt{wang2007robust}) or it can be estimated in a likelihood framework % under the Asymmetric Least Informative distribution
 as illustrated by \cite{bianchi2018estimation}. %Since in the multivariate case no common choices are available ensuring efficiency and robustness for $\widehat{\boldsymbol{\alpha}}_{\tau, \bs{u}}$, %and it will be dependent on the scale of the data. 
 In our multivariate context, \redd{we propose to select the tuning constant} via \red{$K$}-fold cross-validation which allows us to consider $c$ as a data-driven parameter. In particular, for fixed $\tau$ and $\bs u$, we construct a uniform grid of values from $c_{min} = 0.1$ to $c_{max} = \underset{i=1,\dots,n}{\textnormal{max}} \mid \mid \bs Y_i \mid \mid$. Then, for each value of $c \in [c_{min},\dots,c_{max}]$, we fit the proposed model and determine \redd{the optimal value, denoted with $c^\star$, in the sense that it} minimizes the estimated prediction error across the $K$ folds, that is: 
\begin{equation}\label{eq:cvmse}
\widehat{\mbox{CV}}_K (c) = \frac{1}{n} \sum_{k=1}^K \sum_{i \in I_k} (\bs Y_i - \widehat{\bold \theta}^{(k)}_{\tau, \bs u} (c))' (\bs Y_i - \widehat{\bold \theta}^{(k)}_{\tau, \bs u} (c)),
\end{equation}
where $I_1, \dots, I_K$ is a random partition of the $n$ observations into $K$ folds and $\widehat{\bold \theta}^{(k)}_{\tau, \bs u} (c)$ is the estimate of ${\bold \theta}_{\tau, \bs u}$ obtained using the entire sample except data points in the $k$-th fold for a given value of $c$. Finally, $c^\star$ is estimated by:
\begin{equation}\label{eq:cstar}
c^\star = \underset{c \, \in \, [c_{min},\dots,c_{max}]}{\textnormal{arg min}} \widehat{\mbox{CV}}_K (c).
\end{equation}

\subsection{Asymptotic properties}\label{sub:asym}
This section presents the asymptotic properties of the estimator $\widehat{\boldsymbol{\alpha}}_{\tau, \bs{u}}$ in \eqref{eq:ummqpee} \redd{where the $\mathbb{E}[RIF(\bs Y; \nu) \mid {\bs X} = \bs {x}]$ is modeled as a linear function of $\bs X$.} Specifically, we derive the Bahadur-type \citep{bahadur1966note} representation, consistency and asymptotic normality for fixed $\tau$, direction $\bs u$ and $c$. To prove the following results, we follow \cite{firpo2009unconditional} where they consider the IF and not its recentered version. Either using the IF or the RIF, all regression coefficients are the same, the only exception being the intercept. %The proofs of the following results are collected in the \nameref{appendix}.\\% and for the estimator  of the multivariate M-quantile of \cite{kokic2002new}

 Consider the following assumptions:
\begin{enumerate}[label=(\roman*)]
\item[(A1)] The distribution of the random vector $\bs Y$ is absolutely continuous with respect to the Lebesgue measure on $\mathbb{R}^p$, with a density bounded on every compact subset of $\mathbb{R}^p$. % and admits finite second-order moments.
\item[(A2)] The observations $(\bs Y_i, \bs X_i), i = 1,\dots, n$ are an i.i.d. sample from $(\bs Y, \bs X)$.
\item[(A3)] \red{$\mathbb{E}[\mid \mid \eta_{\delta} (\varphi) \Psi(\bs{Y} - \bold \theta) \mid \mid] < \infty, \forall \bold \theta \in \mathbb{R}^p$.}
\item[(A4)] \red{$\mathbb{E}[\mid \mid \eta_{\delta} (\varphi) \Psi(\bs{Y} - \tilde{\bold \theta}) \mid \mid^2] < \infty$ for each $\tilde{\bold \theta}$ in a neighborhood of $\bold \theta_{\tau, \bs u}$.}
%\item[(A3)] The function $\Psi (\cdot)$ in \eqref{eq:mhub} is bounded, non-decreasing in each argument and possesses bounded derivatives up to the second order with the convention $\Psi(\bs 0) = \bs 0$.\medskip
%\item[(A4)] $\mathbb{E}[\mid \mid \eta_{\delta} (\varphi) \Psi(\bs{Y} - \bold \theta) \mid \mid^2] < \infty, \forall \bold \theta \in \mathbb{R}^p$.\medskip
\item[(A5)] The $p \times p$ matrix ${\bs M (\bold \theta_{\tau, \bs{u}})}$ in \eqref{eq:M} is positive definite.
\item[(A6)] $\widehat{\Omega}_{\bs X}$ is nonsingular almost surely for $n$ sufficiently large and converges to $\Omega_{\bs X} = \mathbb{E}[ \bs X \bs X' ]$.
\end{enumerate}
It is worth noticing that assumptions A1-A6 are quite mild and are standard in robust estimation theory. For instance, %assumption A1 holds when $\bs Y$ is a multivariate Normal or Student t with $\nu > 2$ degrees of freedom; instead
 assumptions A1-A5 are needed for the Bahadur (\citealt{bahadur1966note}) representation and ensure the invertibility of ${\bs M (\bold \theta_{\tau, \bs{u}})}$.
 %the existence and positive-definiteness ensure the invertibility of ${\bs M}$ needed for the Bahadur representation.                

%This Section derives, strong consistency, asymptotic normality and Bahadur-type \citep{bahadur1966note} representation results for the sample $\tau$-th order M-quantiles. 
In order to present the asymptotic properties $\widehat{\boldsymbol{\alpha}}_{\tau, \bs{u}}$, we first need to establish the Bahadur-type representation for $\widehat{\bold \theta}_{\tau, \bs{u}}$ and its limiting distribution. 
%Our first established a Bahadur-type representation, consistency and asymptotic normality for the estimator $\bold \theta_{\tau, \bs{u}}$ of the multivariate M-quantile of \cite{kokic2002new}. % (Section \ref{sec:asyM}).

%\subsection{Asymptotic representation of the conditional geometric M-quantile}\label{sec:asyM}
%Let $\bs D_1 (\theta)$ denote a $d \times d$ matrix:
%\begin{equation}\label{eq:D1}
%\bs D_1 (\theta) = \mathbb{E}[\eta_{\delta} (\alpha) \nabla_\theta \Psi(\bs{y} - \theta)].
%\end{equation}

%We now define and obtain an asymptotic Bahadur-type \citep{bahadur1966note} representation of the sample conditional geometric M-quantile, $\widehat{\theta} = \widehat{\theta_{\tau, \bs{u}}}$, of \cite{kokic2002new}.
\begin{theorem}\label{thm:bahac}
Let assumptions A1-A5 hold. Then, for any $\tau \in (0, \frac{1}{2}]$ and $\bs{u} \in \mathcal{S}^{p-1}$, the following asymptotic linear representation holds:
\begin{equation}
\sqrt{n} (\widehat{\bold \theta}_{\tau, \bs{u}} - \bold \theta_{\tau, \bs{u}}) = {\bs M} {(\bold \theta_{\tau, \bs{u}})}^{-1} \frac{1}{\sqrt{n}} \sum_{i=1}^n \eta_{\delta} (\varphi_i) \Psi(\bs{Y}_i - \bold \theta_{\tau, \bs{u}}) + \omicron_p (1)
\end{equation}
and 
\begin{equation}
\sqrt{n} (\widehat{\bold \theta}_{\tau, \bs{u}} - \bold \theta_{\tau, \bs{u}}) \overset{p}{\to} \mathcal{N} (\bs 0, {\bs M} {(\bold \theta_{\tau, \bs{u}})}^{-1} \bs D (\bold \theta_{\tau, \bs{u}}) {\bs M} {(\bold \theta_{\tau, \bs{u}})}^{-1}) \qquad \textnormal{as} \qquad n \rightarrow \infty,
\end{equation}
where $\bs D (\bold \theta_{\tau, \bs{u}})$ defines a $p \times p$ matrix:
\begin{equation}\label{eq:D}
\bs D (\bold \theta_{\tau, \bs{u}}) = \mathbb{E}[\eta^2_{\delta} (\varphi) \Psi(\bs{Y} - \bold \theta_{\tau, \bs{u}}) \Psi' (\bs{Y} - \bold \theta_{\tau, \bs{u}})].
\end{equation}
\end{theorem}
\begin{proof}
See Proof of Theorem \ref{thm:bahac} in the Supplementary Materials.
\end{proof}

%Theorem \ref{thm:bahac} guarantees that the sample geometric M-quantile are $\sqrt{n}$-consistent estimates of the corresponding population M-quantiles and are asymptotically normally distributed.\\

%\subsection{Asymptotic representation of the UMQPE}\label{sec:asyUMQPE}

To prove consistency and asymptotic normality of $\widehat{\boldsymbol{\alpha}}_{\tau, \bs{u}}$, we exploit Theorem \ref{thm:bahac} and define the $k \times p$ matrix of OLS regression coefficients of ${IF} (\bs{Y}; \bold \theta_{\tau, \bs{u}})$ on $\bs X$:
\begin{equation}
\widehat{\bold \beta} (\bold \theta_{\tau, \bs{u}}) = \widehat{\Omega}^{-1}_{\bs X} \frac{1}{n} \sum_{i=1}^n \bs X_i \, IF'(\bs{Y}_i; \bold \theta_{\tau, \bs{u}})
\end{equation}
whose population counterpart is:
\begin{equation}
\bold \beta (\bold \theta_{\tau, \bs{u}}) = \Omega^{-1}_{\bs X} \mathbb{E}[ \bs X \, IF'(\bs{Y}; \bold \theta_{\tau, \bs{u}}) ].
\end{equation}
%Similarly to Section \ref{sub:asyu}, note that we have used the IF and not its recentered version as the dependent variable in the regression: all regression coefficients are the same in both cases, the only exception being the intercept. 
Also, let us denote for $i = 1,\dots,n$,
\begin{equation}
\bold \gamma^\star_i (\bold \theta_{\tau, \bs{u}}) = \textnormal{vec}\Big( \Omega^{-1}_{\bs X} \bs X_i \bs z'_i (\bold \theta_{\tau, \bs{u}}) \Big) \qquad \textnormal{and} \qquad \bs z_i (\bold \theta_{\tau, \bs{u}}) = {IF}(\bs{Y}_i; \bold \theta_{\tau, \bs{u}}) - \bold \beta' (\bold \theta_{\tau, \bs{u}}) \bs X_i.
\end{equation}
where the $\textnormal{vec} (\cdot)$ operator %is a linear transformation which 
converts a matrix into a column vector by stacking its columns on top of one another. Finally, we define:
\begin{equation}
\widehat{\boldsymbol{\alpha}}_{\tau, \bs{u}}^\star = \textnormal{vec}(\widehat{\boldsymbol{\alpha}}_{\tau, \bs{u}}) \qquad \textnormal{and} \qquad \boldsymbol{\alpha}_{\tau, \bs{u}}^\star = \textnormal{vec}(\boldsymbol{\alpha}_{\tau, \bs{u}}).
\end{equation}

%We are now ready to state the limiting distribution for $\widehat{UMMQPE}_{\tau, \bs{u}}^\star$.
\begin{theorem}\label{thm:mUMQPE}
Let assumptions A1-A6 hold. Then, for any $\tau \in (0, \frac{1}{2}]$ and $\bs{u} \in \mathcal{S}^{p-1}$, the following asymptotic linear representation holds:
\begin{equation}\label{eq:UMQPEbahac}
\sqrt{n} (\widehat{\boldsymbol{\alpha}}_{\tau, \bs{u}}^\star - {\boldsymbol{\alpha}}_{\tau, \bs{u}}^\star) = \frac{1}{\sqrt{n}} \sum_{i=1}^n \bs S_i (\bold \theta_{\tau, \bs{u}}) + \omicron_p (1)
\end{equation}
and 
\begin{equation}\label{eq:UMQPEn}
\sqrt{n} (\widehat{\boldsymbol{\alpha}}_{\tau, \bs{u}}^\star - {\boldsymbol{\alpha}}_{\tau, \bs{u}}^\star) \overset{p}{\to} \mathcal{N} \Big(\bs 0, \mathbb{E}\Big[\bs S (\bold \theta_{\tau, \bs{u}}) \bs S' (\bold \theta_{\tau, \bs{u}})\Big]\Big) \qquad \textnormal{as} \qquad n \rightarrow \infty,
\end{equation}
where $\bs S_i (\bold \theta_{\tau, \bs{u}})$ is a $kp$-dimensional vector:
\begin{equation}
\bs S_i (\bold \theta_{\tau, \bs{u}}) = \nabla_{\bold \theta_{\tau, \bs{u}}} \bold \beta^\star (\bold \theta) \bs M (\bold \theta_{\tau, \bs{u}})^{-1} \eta_{\delta} (\varphi_i) \Psi(\bs{Y}_i - \bold \theta_{\tau, \bs{u}}) + \bold \gamma^\star_i (\bold \theta_{\tau, \bs{u}}),
\end{equation}
and $\nabla_{\bold \theta_{\tau, \bs{u}}} \bold \beta^\star (\bold \theta)$ is the derivative of $\bold \beta^\star (\bold \theta)$ with respect to $\bold \theta$ evaluated at $\bold \theta_{\tau, \bs{u}}$.
 %and
%\begin{equation}
%\gamma^\star_i (\theta_{\tau, \bs{u}}) = \textnormal{vec}\Big( \Omega^{-1}_{\bs X} \bs X_i \bs z^\top_i (\theta) \Big) \qquad \textnormal{and} \qquad \bs z_i (\theta) = {IF}(\bs{y}_i; \theta) - \beta^\top (\theta) \bs X_i.
%\end{equation}
\end{theorem}
\begin{proof}
See Proof of Theorem \ref{thm:mUMQPE} in the Supplementary Materials.
\end{proof}
%The proof of the previous Theorem follows closely the one illustrated in Theorem \ref{thm:uUMQPE}.

%Under the assumptions of Theorem \ref{thm:mUMQPE}, we are now ready to establish a limiting distribution for a suitable normalized version of the RIF-OLS estimator.
%\begin{theorem}\label{thm:UMQPEn}
%For any $\tau \in (0, \frac{1}{2})$ and $\bs{u}$ in the $(d-1)$-th dimensional sphere, under Assumptions (A1)-(A5), the asymptotic distribution of the centered and normalized UMQPE is:
%\begin{equation}\label{eq:UMQPEn}
%\sqrt{N} (\widehat{\textnormal{UMQPE}_{\tau, \bs{u}}^\star} - \textnormal{UMQPE}_{\tau, \bs{u}}^\star) \overset{p}{\to} \mathcal{N} \Big(\bs 0, \mathbb{E}\Big[\bs S (\theta_{\tau, \bs{u}}) \bs S^\top (\theta_{\tau, \bs{u}})\Big]\Big) \qquad \textnormal{as} \qquad N \rightarrow \infty,
%\end{equation}
%where $\bs S (\theta_{\tau, \bs{u}})$ is a $kd$-dimensional vector:
%\begin{equation}
%\bs S (\theta_{\tau, \bs{u}}) = \Big( \nabla \beta^\star (\theta) \bs D^{-1}_{1} (\theta_{\tau, \bs{u}}) \bs D_2 (\theta_{\tau, \bs{u}}) + \gamma^\star (\theta) \Big).
%\end{equation}
%\end{theorem}

%\begin{proof}
%Under the Assumptions of Theorem \ref{thm:bahac}, the asymptotic normality of the sample geometric M-quantile follows from the Slutsky's Theorem and the multivariate Central Limit Theorem. 
%\end{proof}

\red{In addition, when multiple directions are of interest, we may estimate the proposed model at different directions simultaneously to gain better estimation accuracy by incorporating the associations among the considered directions. Hence, by proceeding as in the proof of Theorem \ref{thm:mUMQPE} and applying the multivariate Central Limit Theorem to the Bahadur representation in \eqref{eq:UMQPEbahac}, we derive the asymptotic distribution of the UMQPE estimator when multiple directions are considered jointly, as shown in the following remark.
\begin{remark}
%be a finite set of directions in,
Let $\{ \bs u_1, \dots, \bs u_{J} \} \subset \mathcal{S}^{p-1}$, with $J$ being a fixed positive integer. Suppose assumptions A1-A6 hold, then the joint asymptotic distribution of $\sqrt{n} \big( \widehat{\boldsymbol{\alpha}}_{\tau, \bs u_1}^\star - {\boldsymbol{\alpha}}_{\tau, \bs u_1}^\star, \dots, \widehat{\boldsymbol{\alpha}}_{\tau, \bs u_J}^\star - {\boldsymbol{\alpha}}_{\tau, \bs u_J}^\star \big)$
%\begin{equation}\label{eq:asr1}
%\sqrt{d} \bigg( \widehat{\boldsymbol \beta}_{MMQ} (\tau, \bs u_1) - \boldsymbol \beta (\tau, \bs u_1), \dots, \widehat{\boldsymbol \beta}_{MMQ} (\tau, \bs u_J) - \boldsymbol \beta (\tau, \bs u_J) \bigg)
%\end{equation}
is Gaussian with zero mean and the asymptotic covariance matrix between $\sqrt{n} (\widehat{\boldsymbol{\alpha}}_{\tau, \bs u_r}^\star - {\boldsymbol{\alpha}}_{\tau, \bs u_r}^\star)$ and $\sqrt{n} (\widehat{\boldsymbol{\alpha}}_{\tau, \bs u_s}^\star - {\boldsymbol{\alpha}}_{\tau, \bs u_s}^\star)$, where $1 \leq r, s \leq J$, will be given by $\mathbb{E}\Big[\bs S (\bold \theta_{\tau, \bs u_r}) \bs S' (\bold \theta_{\tau, \bs u_s})\Big]$.
%\begin{equation}\label{eq:asr2}
%{\bs H (\bs u_r)}^{-1} {\bs B (\bs u_r, \bs u_s)} {\bs H (\bs u_s)}^{-1},
%\end{equation}
% \quad \textnormal{as} \quad d \rightarrow \infty
\end{remark}
}

%the previous
In order to use Theorem \ref{thm:mUMQPE} to build confidence intervals and hypothesis tests, in what follows we provide a consistent estimator of the asymptotic covariance matrix of $\widehat{{\boldsymbol{\alpha}}}_{\tau, \bs{u}}^\star$ in \eqref{eq:UMQPEn}. The analytical form of the asymptotic covariance matrix suggests the following estimator:
\begin{equation}\label{eq:svUMQPE}
\widehat{\bs V} (\widehat{\bold \theta}_{\tau, \bs{u}}) =  \frac{1}{n} \sum_{i=1}^n \widehat{\bs S}_i (\widehat{\bold \theta}_{\tau, \bs{u}}) \widehat{\bs S}'_i (\widehat{\bold \theta}_{\tau, \bs{u}})
\end{equation}
where $\widehat{\bs S}_i (\widehat{\bold \theta}_{\tau, \bs{u}})$ is the $kp$-dimensional vector:
\begin{equation}
\widehat{\bs S}_i (\widehat{\bold \theta}_{\tau, \bs{u}}) = \Big( \nabla_{\widehat{\bold \theta}_{\tau, \bs{u}}} \widehat{\bold \beta}^\star ({\bold \theta}) \bs M (\widehat{\bold \theta}_{\tau, \bs{u}})^{-1} \eta_{\delta} (\widehat{\varphi}_i) \Psi(\bs{Y}_i - \widehat{\bold \theta}_{\tau, \bs{u}}) + \widehat{\bold \gamma}_i^\star (\widehat{\bold \theta}_{\tau, \bs{u}}) \Big)
\end{equation}
and $\nabla_{\widehat{\bold \theta}_{\tau, \bs{u}}} \widehat{\bold \beta}^\star ({\bold \theta})$ can be obtained via numerical differentiation.
%we approximate this term by using the R package numDeriv.

\red{In order to establish consistency of the estimator in \eqref{eq:svUMQPE} we impose the following additional assumption.
\begin{enumerate}[label=(\roman*)]
\item[(A7)] $\bs S (\bold \theta_{\tau, \bs{u}})$ is continuous at $\bold \theta_{\tau, \bs u}$ and \redd{there exists a neighborhood of $\bold \theta_{\tau, \bs u}$, $\mathcal{I}$, such that for any $\tilde{\bold \theta}$ in this region, it holds that $\mathbb{E}[\underset{\small \tilde{\bold \theta} \in \mathcal{I}}{\sup} \Vert \bs S (\tilde{\bold \theta}) \bs S' (\tilde{\bold \theta}) \Vert] < \infty$.}
%$\bs S (\bold \theta_{\tau, \bs{u}})$ is continuous at $\bold \theta_{\tau, \bs u}$ and $\mathbb{E}[\sup \Vert \bs S (\tilde{\bold \theta}) \bs S' (\tilde{\bold \theta}) \Vert] < \infty$ for each $\tilde{\bold \theta}$ in a neighborhood of $\bold \theta_{\tau, \bs u}$.
\end{enumerate}
%The next Theorem proves consistency of $\widehat{\bs V} (\widehat{\bold \theta}_{\tau, \bs{u}})$.
Then, the next Theorem proves consistency of $\widehat{\bs V} (\widehat{\bold \theta}_{\tau, \bs{u}})$.
\begin{theorem}\label{thm:vmUMQPE}
Let assumptions A1-A7 hold, 
\begin{equation}
\widehat{\bs V} (\widehat{\bold \theta}_{\tau, \bs{u}}) - \mathbb{E}\Big[\bs S (\bold \theta_{\tau, \bs{u}}) \bs S' (\bold \theta_{\tau, \bs{u}})\Big] \overset{p}{\to} 0,
\end{equation}
where the notation is understood to indicate convergence of the matrices element by element.
\end{theorem}
\begin{proof}
See Proof of Theorem \ref{thm:vmUMQPE} in the Supplementary Materials.
\end{proof}
}

\section{Application}\label{sec:app}

In this section, we consider data from the Survey on Household Income and Wealth (SHIW) 2016 conducted by the Bank of Italy to show the relevance of our methodology. We analyze the impact of economic and socio-demographic factors on households wealth and consumption levels, accounting both for the presence of outliers and the correlation structure between the two outcomes. We are interested in evaluating whether these effects are more pronounced on more disadvantaged families than on richer ones. In this setting, the limitation of using a conditional (M-)quantile model is that the effect of the covariates at different quantile levels may be masked by the set of conditioning variables, i.e., the characteristics of the family. Once we have conditioned on the explanatory variables, for instance, %including $\bs X$ in a conditional model implies that 
 %observations at the center of the unconditional distribution of the responses may potentially be at the boundary of the conditional distribution, and viceversa,
 the $10$-th (M-)quantile of the unconditional distribution of the responses may potentially be very different from the $10$-th (M-)quantile of the conditional distribution,
 so the coefficients of conditional (M-)quantile regression cannot be interpreted as unconditional effects.
%Existing works in the literature have typically focused on conditional univariate regression approaches, hence neglecting the dependence structure between wealth and consumption. Also, conditioning on the covariates $\bs X$ in the model implies that observations at the center of the distribution of $\bs Y$ are potentially at the boundary of the conditional distribution, and viceversa, so the coefficients of conditional regressions cannot be interpreted as the unconditional effect.   
 On the other hand, by using our unconditional method, %is designed to isolate the effect of the explanatory variables by integrating out the RIF over the conditioning variables. As a result, %Also, by looking at the unconditional distribution of the responses, 
 the UMQPE provides an estimate of the impact of covariates across the entire population and not merely among population subgroups, consisting of families who share the same values of the included covariates. %, as it is done in conditional inference. 
 In what follows, we fit the proposed regression method at different points of the unconditional distributions of family wealth and consumption, and illustrate the difference between the conditional and unconditional approaches. %(M-)quantile regression models.

\subsection{Data description}\label{sub:data}
The SHIW (\href{https://www.bancaditalia.it/statistiche/tematiche/indagini-famiglie-imprese/bilanci-famiglie/distribuzione-microdati/index.html}{\nolinkurl{https://www.bancaditalia.it}}) is an annual survey conducted by the Bank of Italy whose aims are to provide information on the economic and financial behaviours of Italian households and collect reliable, comparable and representative data of the population resident in Italy. %The survey was begun in the 1960s to gather data on the incomes and savings of Italian households. 
 %Over the years, it has grown in scope and now allows to study the evolution of income, consumption and wealth given demographic changes over the years. 
 This survey is widely regarded as the basis of the most reliable estimates for macroeconomics studies. %; see, among others, the works of \cite{battistin2003we, fiorio2005workers, paiella2007does} and \cite{jappelli2010}. %The SHIW data is publicly available at \href{https://www.bancaditalia.it/statistiche/tematiche/indagini-famiglie-imprese/bilanci-famiglie/distribuzione-microdati/index.html}{\nolinkurl{https://www.bancaditalia.it}} and 
 The sample is drawn in two stages, the primary and secondary sampling units are municipalities and households, respectively. Before the primary units are selected, they are stratified by region and population size. Data are collected mainly via an electronic questionnaire using the Computer Assisted Personal Interviewing program while the remaining interviews are conducted using the Paper And Pencil Personal Interviewing program.
 %The questionnaire used in each year is available at \href{https://www.bancaditalia.it/statistiche/tematiche/indagini-famiglie-imprese/bilanci-famiglie/documentazione/index.html}{\nolinkurl{https://www.bancaditalia/Documentation for the microdata.it}}
%paper-based questionnaires
%The questionnaire has a modular structure and it is composed of a general part addressing aspects relevant to all households and a series of additional sections containing questions relevant to specific subsets of households.
%The questionnaire used in each year is fully available on the web page ?Documentation for the microdata? of the section dedicated to the survey of the Bank of Italy website.
%The 2016 survey comprises 7420 households composed of 16462 individuals. 
 
 %The SHIW contains detailed information on household's characteristics and their respective heads. 
 In this work, following established custom we transform the dependent variables, i.e., household consumption (LCON) and net wealth (LWEA) to natural logarithm. In particular, LCON is defined as the sum of household's expenditure on durables and non-durable goods while net wealth is obtained as the algebraic sum of real assets, financial assets and financial liabilities. The set of considered predictors includes the log of net disposable income (LINC), defined as the sum of payroll income, pensions, net transfers, net self-employment income and property income sources, and relevant information on the household's head such as age (Age) and age squared (Age2) measured in years. Also, gender (male (baseline)), marital (married (baseline), never married, separated, widowed), employment status (employee (baseline), self-employed, not-employed) and educational level (elementary (baseline), middle, vocational, high school, university or higher) are included as dummy variables. Finally, a categorical variable is included to investigate the presence of regional divergences in wealth and consumption levels depending on the region of residence (north (baseline), centre, south and islands). Table S7 in the Supplementary Materials summarizes the descriptive statistics of the included variables. The considered sample contains 6802 households. %After removing missing data %and all amounts are expressed in euros
 
%Cannari et al. (1995) and Marenzi (1996) considered years 1989 and 1991 r
%To estimate of tax evasion a measure of tax evasion
%It asks respondents a very broad range of questions
%income inequality
%SHIW provides information on consumption and income from 1980 to 2006, and includes a sizable panel component that allows econometricians to estimate income, consumption and wealth processes and to analyze labor market and portfolio transitions. Given the population changes associated with the demographic transition and policy reforms that have occurred in Italy over the last two decades, the data represent an ideal context for applied macroeconomic study. 
%The objectives of the SHIW are the study of the economic and financial behaviours of Italian households and the collection of reliable, comparable and representative data of the population. %resident in Italy.\\% Among its main objectives The dataset contains detailed information on household composition, age, education, labour market variables, income, savings and consumption. %
%households' economic and financial behaviour
%The main objective of the SHIW is to  
%un'indagine sui bilanci familiari con l'obiettivo di raccogliere dati sulla distribuzione del reddito e sulle caratteristiche strutturali del risparmio;
%The survey was begun in the 1960s to gather data on the incomes and savings of Italian households. Over the years, it has grown in scope and now includes wealth and other aspects of households' economic and financial behaviour.
  
%the scatter plot, the histograms
 As a preliminary step, we study the unconditional distributions of households wealth and consumption. The histograms of LWEA and LCON unconditional distributions in Figure S5 of the Supplementary Materials \red{reveal that, while normality seems tenable for LCON, there are potentially influential observations in the distribution of LWEA and indicate a departure from the Gaussian assumption,} %reveal the presence of potentially influential observations in the data and indicate a departure from the Gaussian assumption, 
 having fat tails and pronounced asymmetries. Furthermore, the empirical correlation between LCON and LWEA equals to 0.473. As expected, consumption and wealth are positively correlated, justifying the need for a multivariate approach that considers these two dimensions together. \red{The discreteness of LWEA that shows up especially at low values of wealth is due to the rounding effect of the interviews towards round figures (\citealt{groves2011survey}).} 
%justifying the joint modeling approach we adopt in this paper.
%In doing so, it highlights the importance of considering income, consumption and wealth together as part of a conceptual framework.
%It has also stressed the need for a multi-dimensional approach that considers these three dimensions together, in order to gain a more complete understanding of household economic well-being
 Consequently, the presented unconditional regression model is appropriate to account for outlying observations and investigate how the relationship between responses and explanatory variables can vary across the unconditional distribution of family wealth and consumption.

%The lack of normality for the model residuals is most probably due to several outliers in regions. 
%A categorical variable is included to investigate the presence of regional differences between northern, central and southern Italy
%The characteristics include gender of the household head captured as a dummy (males are adopted as the reference category), age and the square of age measured in years. Household human capital is measured by the highest educational attainment of the household head. Educational attainment is classified into no formal education, basic, secondary, tertiary and other education. These categories are included as dummies with no formal education as the reference category. The marital and employment status of household heads are also included as dummies. Marital status is captured as whether the head of household is currently in a marital or consensual union. Employment status refers to whether the household head is employed in a formal wage sector or informal self-employment.

\subsection{Modeling household wealth and consumption}\label{sub:mul} 
% %This empirical application aims to jointly model households' log- wealth, LWEA, and consumption, LCON, using the same explanatory variables of Application \ref{sub:uni}.
% However, a by-product of the proposed multivariate approach is the possibility to study the dependence
We analyze the SHIW 2016 data to jointly model households' log-wealth, LWEA, and log-consumption, LCON, as a function of the predictors in Table S7. \red{We fit the proposed approach at levels $\tau = 0.10$, $\tau = 0.50$ and $\tau = 0.90$, which can estimated by simply noting that for $\tau \in (0, \frac{1}{2}]$, $\bold \theta_{1 - \tau, \bs{u}} = \bold \theta_{\tau, -\bs{u}}$ (see \citealt{kokic2002new}). %\footnote{As stated in , to estimate the 90-th multivariate quantile, M-quantile and expectile, we exploit the relationship $\bold \theta_{1 - \tau, \bs{u}} = \bold \theta_{\tau, -\bs{u}}$, $\tau \in (0, \frac{1}{2}]$. To facilitate interpretation of the results, we will continue to use the notation $\tau = 0.90$ in this section.}
 %That is, the space is divided into two regions by a hyper-plane orthogonal to $\bs u$ and residuals in the two half-spaces receive asymmetric weights according to the function $\eta_{\delta} (\alpha)$ related to \eqref{eq:min}.
 As explained in Section \ref{sec:ummq}, a meaningful direction $\bs u$ shall be determined for ordering multivariate observations, taking into account the problem under consideration. 
 %In this paper, we propose to choose $\bs u$ empirically taking into account the problem at hand. 
%Specifically, in economics and finance, there has been by now a long tradition of studies investigating the relationship between consumption and wealth (see \citealt{campbell1989consumption, deaton1992understanding} and \citealt{campbell1993consumption} for example). In this literature, a central object is the wealth-consumption ratio (\citealt{li2005wealth, lustig2013wealth}). %which plays a fundamental role in long run risk models (\citealt{koijen2010long}) and it emerges as a strong forecaster of asset returns (\citealt{lettau2001consumption, li2005wealth, rudd2006empirical}). 
 %In our specific application where $\mathbf{Y} = (Y_{(1)}, Y_{(2)})'$ denote the multivariate random vector collecting households' log-wealth and log-consumption, it is immediate to see that $Z = \log \Big( \frac{Y^{u_1}_{(1)}}{Y^{-u_2}_{(2)}} \Big)$ is the log wealth-consumption ratio if $\bs u = (1, - 1)'$. Following the definition in \eqref{eq:min}, we can justify this choice from a practical point of view as each residual, $\bs Y - \bold \theta_{\tau, \bs{u}}$, is assigned a weight depending on the ratio between wealth and consumption. 
 A possible way to do this is to observe that ratios of the type $Z = \log \Big( \frac{Y^{u_1}_1}{Y^{-u_2}_2} \Big) = \bs u' (\log Y_1, \log Y_2)$ bears great importance in several fields like finance, economics and growth models, with $Y_1$ and $Y_2$ representing wealth and consumption at the household level, respectively. In order to explicitly account for the existing positive correlation between wealth and consumption, we use the direction $\bs u = (\frac{1}{\sqrt{2}},\frac{1}{\sqrt{2}})'$ throughout the rest of the section (the analysis with $\bs u = (\frac{1}{\sqrt{2}},-\frac{1}{\sqrt{2}})'$ is shown in the Supplementary Materials). %In order to explicitly account for the existing link between wealth and consumption in the choice of $\bs u$, we use this direction throughout the rest of the section. 
 To give a graphical intuition on the use of such direction, Figure S6 in the Supplementary Materials shows the scatter plots of family wealth and consumption where $\bs u$ is represented by the green arrow, while the blue points denote the multivariate quantile ${\bold \theta}_{\tau, \bs u}$ with $c=0$ at $\tau = 0.10$ (left) and $\tau = 0.90$ (left) as an example. The black dashed lines correspond to the hyperplanes orthogonal to $\bs u$ and passing through ${\bold \theta}_{\tau, \bs u}$ where datapoints in the upper half-plane are shown in red while those in the lower half-plane are shown in gray. In doing so, the sample space is divided into two regions by the hyperplane orthogonal to the selected direction %which in our specific application is and passing through ${\bold \theta}_{\tau, \bs u}$, 
 and any observation in the region in direction $\bs u$ from ${\bold \theta}_{\tau, \bs u}$ receives a certain weight, while observations in the opposite direction receive a different weight depending both on their length and the angle they form with $\bs u$, according to the weighting function $\eta_\delta (\varphi)$ below \eqref{eq:min}. %so $\cos \varphi = \frac{(\bs{Y} - \bold \theta)' { \bs u}}{\Vert \bs{Y} - \bold \theta \Vert \Vert { \bs u} \Vert}$ %As one can see, the considered directional approach allows us to give a difference importance to each variable involved in the analysis.
 %As one can see, when the direction $\bs u = (\frac{1}{\sqrt{2}},\frac{1}{\sqrt{2}})'$ is considered, observations are essentially divided in two groups. 
 In particular, at $\tau = 0.10$, households below the half-plane generally have low or moderate levels of consumption and wealth, and a certain weight will be assigned to them compared to wealthier and high spender families located in the other region. When $\tau = 0.90$, the multivariate quantile shifts upwards to the right and thus also the separating hyperplane so that it distinguishes families with jointly high consumption and wealth patterns who will be weighted differently than all those below as shown in gray.}
  
 The estimation of the optimal tuning constant $c^\star$ is obtained using a 5-fold cross-validation over a sequence of 200 possible values as described in Section \ref{sub:est}. The results are reported in Table \ref{tab:coefmcstar} ($c = c^\star$) where we compare the proposed unconditional regression with the standard conditional regression approach. \red{In particular, the left-hand side columns (labeled as Unconditional Regression) refer to the estimates of the UMQPE, $\boldsymbol{\alpha}_{\tau, \bs{u}}$, obtained as described in \eqref{eq:ummqpee}. Conversely, the right-hand columns (labeled as Conditional Regression) refer to the estimates of the coefficients $\tilde{\bold \beta}_{\tau, \bs{u}}$ of a multivariate linear conditional M-quantile regression, i.e., $\boldsymbol{\theta}_{\tau, \bs{u}} (\bs x) = \tilde{\bold \beta}'_{\tau, \bs{u}} \bs x$, obtained using the IRLS algorithm mentioned in Section \ref{sec:ummq}.} %The results are reported in Table \ref{tab:coefmcstar} ($c = c^\star$) where we compare the proposed unconditional regression with the standard conditional regression approach.
 To display the sampling variation, asymptotic standard errors obtained using the results in Section \ref{sub:asym} are presented in parentheses. Parameter estimates are displayed in boldface when significant at the standard 5\% level. 
\\
The main findings can be summarised as follows. The selected values of $c$ are \red{1.373, 10.966 and 10.378 at level 0.10, 0.50 and 0.90, respectively. This implies that the estimates are close to the quantile case at $\tau = 0.10$ as more than 77\% of residuals are down-weighted (Huberised) meanwhile, at $\tau = 0.50$ and $\tau = 0.90$ correspond to expectile estimation as no observations are Huberised.} The estimated values for $c$ reflect the negatively skewed distribution of wealth and support the exploratory analysis in Figure S5. %it is noteworthy that the percentage of down-weighted (Huberised) residuals by the selected tuning constant $c^\star = 3.648$ is equal to 10.901, 9.364 and 3.287 for the GUQR %, and it corresponds to 2.875, 2.747 and 3.074 for the conditional MQR 
%Even the log-normal distribution with its extreme skewness still maintains an appropriate looking transition for the cMAD and MM approaches.
 %at the three $\tau$ levels, revealing that a small proportions of residuals are bounded by $c^\star$. %The proportions of Huberised residuals reveal that few observations are bounded by $c$ but they remain constant across $\tau$, denoting an adequate degree of robustness for all considered $\tau$ levels. 
\\
Because the selected values of $c$ lead towards expectile estimation, especially above $\tau = 0.50$, we also consider the case $c=0$ (see Table S8 of the Supplementary Materials) which allows us to estimate the impact of the covariates on the unconditional multivariate quantiles of the response. Comparing Tables \ref{tab:coefmcstar} and S8, 
 %Moving on to Table \ref{tab:coefmc0}, 
 %the estimates differ substantially at $\tau = 0.10$, while they are very similar from the median level upwards ($\tau = 0.50$ and $\tau = 0.90$).
 slight differences in terms of estimation can be found. This is attributable to the choice of $c$ as the two models allow to target different population parameters by selecting different values for the tuning constant of the Huber function. The above points demonstrate the flexibility of the methodology proposed to extend the classical OLS regression for assessing the effect of covariates, not only at the center, but also at different parts of the unconditional distribution of interest.
\\
%Such differences are attributable 
%Such differences are attributable to the choice of $c$ as %: for $c \geq c^\star$, few observations are down-weighted hence, the results are close with those of classical mean regression; %, which are not reported here to safe space instead, 
%for $c=0$ all observations are Huberised to provide the necessary resistance against outliers. 
%methodology proposed allows to target different population parameters
%The methodology proposed can be viewed as an alternative to the LQMM that was proposed by Geraci and Bottai (2014), although the two models target different population parameters.
%Why consider MQs when one key advantage of quantile regression is the more intuitive interpretations? Although not the same, both quantile and MQ models target essentially the same part of the distribution of interest.
%can be applied to estimate different target parameters
%Consistently with the quantile framework, the estimated intercepts increase when moving from lower to upper quantiles.
Nevertheless, there are still similar results between Tables \ref{tab:coefmcstar} and S8. Point estimates generally increase in magnitude when moving outward from the bulk of the data. Income elasticity is positively associated with both wealth and consumption for all investigated $\tau$ levels. One can see that there are small differences in consumption expenditure among males and females. Moreover, education, marital and employment status are important determinants of family's consumption and wealth levels, with an increasing trend as $\tau \rightarrow 0$. There also appears to be significant regional disparities across the distributions of the responses as southern regions and islands generally have lower levels of wealth and consumption. %with respect to families in the centre and north of Italy. 
 It is important to note that the effect of education and marital status between the conditional and unconditional models is very different, especially at $c=0$. As shown in Theorem \ref{thm:wavem}, this corresponds to the case where large differences exist between the UMQPE and the CMQPE. This may be due to the fact that the matching, $s_{\tau, \bs{u}} (\bs X)$, and weighting, $\bs W_{\tau, \bs{u}} (\bs X)$, functions vary across the values of $\bs X$, which means that the conditional effects do not average up to their respective unconditional effects. By contrast, the conditional and unconditional models provide similar estimates for \red{LINC, age and employment status at $c = c^\star$ in particular, suggesting that $s_{\tau, \bs u} (\bs X) \approx \tau$ does not vary very much for all values of $\bs X$. In this case, we have that $\boldsymbol{\alpha}_{\tau, \bs{u}} = \mathbb{E}[\bs W_{\tau, \bs{u}} (\bs X) \, \boldsymbol{\alpha}_{s_{\tau, \bs u} (\bs X), \bs u} (\bs X)] \approx \mathbb{E}[\bs W_{\tau, \bs{u}} (\bs X) \, \boldsymbol{\alpha}_{\tau, \bs u} (\bs X)]$. Under the considered linear model targeting the conditional multivariate M-quantile $\boldsymbol{\theta}_{\tau, \bs{u}} (\bs x)$, the CMQPE, $\boldsymbol{\alpha}_{\tau, \bs{u}} (\bs x) = \tilde{\bold \beta}_{\tau, \bs{u}}$, which implies that $\boldsymbol{\alpha}_{\tau, \bs{u}} \approx \tilde{\bold \beta}_{\tau, \bs{u}}$.} %suggesting that both functions are relatively constant for all values of $\bs X$, and that the conditional and unconditional distributions might be similar 
\\
As a means of further comparison, we analyze the log-levels of wealth and consumption with the UQR of \cite{firpo2009unconditional} by fitting two univariate models independently. Table S9 of the Supplementary Materials reports the corresponding parameter estimates and standard errors at the examined $\tau$ levels.
%Table \ref{tab:uqr} reports the results from the univariate UQR of \cite{firpo2009unconditional} 
%In addition to the proposed model, we compare our methodology with two well-known univariate alternatives
%he LQMM of Geraci & Bottai (2014) with time-constant random intercepts, at quantile level ? = (0.25, 0.50, 0.75, 0.90). Specifically, the two models are estimated on SDQ Int and SDQ Ext scores independently.
%With respect to the parameter estimates of  in Table \ref{tab:uqr}, 
 We can observe that the results are generally in line with the \red{estimates from the proposed multivariate unconditional quantile regression in Table S8. However, some differences can be identified at $\tau = 0.10$} due to the selected direction $\bs u$ and the fact that the univariate UQR completely disregards the dependence between wealth and consumption. By contrast, the proposed model allows to study the direction and magnitude of such correlation at different levels $\tau$. In particular, using \eqref{eq:cor} we represent in Tables \ref{tab:coefmcstar} and S8 the estimated correlation coefficient, $r_{12}$, which indicates that consumption and wealth are strongly correlated with each other and this association slightly decreases for households at the upper end of the responses distribution. \red{At $\tau = 0.10$, the estimated coefficients (0.324 and 0.308) suggest that low-consumption households are likely to be accompanied by low wealth, with an increasing pattern as we move towards the $\tau = 0.50$ and $\tau = 0.90$ levels.}
%Generally, in a given society at a given time, income is positively related to reported subjective well-being, so that individuals with a higher income tend to report higher subjective well-being than those with a lower income.
%amongst households with high levels of resources
%higher for households at the upper end of the distribution
%Households with reserves of wealth can also utilise these to generate income and to support a higher standard of living.
%households with low material wellbeing or of considered to be at risk of poverty
%Both of these reports also emphasised the need to look at the distribution of economic well-being across individuals. This is especially important when there are disparities in achievements across population groups and when these are correlated across dimensions (e.g. when the likelihood of earning a low income is correlated with low educational achievement, poor health status, poor housing, etc.).
%Measures of low income, particularly when accompanied by low wealth, are also important, as low-income people typically experience deprivations in several domains, not
%positive correlation between
%gives a measure of tail correlation
\\
\redd{We conclude the analysis by using the nonparametric estimator of the UMQPE described in Section \ref{sub:est} to take into account possible nonlinear effects of the covariates. That is, we model the $\mathbb{E}[RIF(\bs Y; \bold \theta_{\tau, \bs u}) \mid \bs X]$ including cubic B-splines specified for income and age. To test the linearity assumption between the RIF and the covariates, we conduct a goodness-of-fit test based on the Pillai's Trace test statistic (\citealt{rencher2012methods}, see Table S11 in the Supplementary Materials). The results provides evidence that the model with B-splines fits better than the linear specification for the RIF regression. %To determine whether regression splines fits as well as the linear specification for the RIF,
 In Table \ref{tab:coefmcstar_spline} we thus show the estimated UMQPEs using the nonparametric method at $\tau = 0.10$, $\tau = 0.50$ and $\tau = 0.90$ for the optimal tuning constant $c^\star$, where 95\% confidence intervals (in parentheses) are obtained via nonparametric bootstrap (\citealt{efron1994introduction}). %Parameter estimates are displayed in boldface when significant at the standard 5\% level. 
 By looking at the results, the estimated UMQPEs are similar, in terms of both sign and magnitude, to the estimates in Table \ref{tab:coefmcstar} where the RIF is modeled as a linear function of the covariates.}
%By comparing the estimated UMQPEs in Table \ref{tab:coefmcstar}, the estimates from the latter are very similar to the estimates of the former in terms of both sign and magnitude across the considered $\tau$ levels. This suggests that, at least for this particular application, the potential presence of non-linearity does not seem to have so much impact on the UMQPE estimates.

%ng a linear probability model or a logit gives very similar ave marginal effects. More importantly, Figure 1 shows that the RIF-NP estim are also very similar to the estimates obtained using these two simpler me ods.23 This suggests that, at least for this particular application, using a sim linear specification for the unconditional quantile regressions provides fa accurate estimates of the UQPE. The small difference between RIF-OLS and RIF-NP estimates stands in sharp contrast to the large differences between the RIF-OLS estimates and the conditional quantile regression estimates in panel A.

\redd{Finally, in Section 6 of the Supplementary Materials we provide a graphical representation of how well the considered approach can approximate the effect of more general changes in distribution of the covariates such as those contemplated in the policy effect (\citealt{RePEc:nbr:nberte:0339}).}

\section{Conclusions}\label{sec:con}
%The marginal inferences from the binary part of the two-part model suggest that the most important determinants of people?s willingness to spend on health services are related to socio-cultural and economic factors:
% and conveniently reduces for scalar response varialbe 
%Dependence between multiple response variables of interest is very common, and occurs in crosssectional, case-control, and longitudinal studies. The method presented in this paper allows to investigate the conditional correlation structure of a multivariate response, by analyzing the signs of the residuals from univariate quantile regression models
Extending the univariate work of \cite{firpo2009unconditional}, this paper proposes a unified approach to model the entire unconditional distribution of a multivariate response variable in a regression setting. We make several contributions to the literature. First, by employing the multidimensional Huber's function in \cite{hampel2011robust} we are able to build a comprehensive modeling framework to estimate multivariate unconditional quantiles, M-quantiles and expectiles, choosing the tuning constant in an appropriate manner. Second, in contrast to univariate methods, our multivariate model accounts for the, potentially asymmetric, association structure between the outcome variables. Third, the proposed methodology is easy to implement through an OLS regression of the RIF on the explanatory variables. From a theoretical standpoint, we show that the introduced estimators are consistent, asymptotically normal and can be written as a weighted average %, over the distribution of the covariates, 
 of conditional effects. In addition, we propose a data-driven procedure based on cross-validation to select the optimal tuning constant for estimating the UMQPE \red{and the policy effect}. %we propose an automatic data-driven procedure based on a 5-fold cross-validation for estimating the tuning constant of the Huber's M-function. we propose a data-driven method based on cross-validation for selecting the optimal tuning constant that can adjust its value in scenarios with outliers 
 Finally, we contribute to the empirical literature by analyzing log-levels of wealth and consumption of Italian households collected in the SHIW 2016 data.

%A procedure based on a 5-fold cross-validation is also proposed as an automatic data-driven method for estimating the tuning constant of the Huber's M-function. Simulation studies and an application on the SHIW 2016 are given in order to evaluate the behaviour of the proposed methodology using artificial and real data.
%Correlated, subject-specific, random effects are used to account for dependence within the same response and association between responses observed on the same subject.
%First, it considers a potentially varying strength in the association between the analysed outcomes, as the estimate of the random effect distribution is (M-)quantile-specific. This could be particularly appropriate to those phenomena that experience so-called tail dependences
%To implement the methodology, we consider the multidimensional Huber function in \cite{hampel2011robust} which \red{allows us to build a single modeling framework} for

%From a methodological standpoint, possible future developments are as follows%
Possible future developments of this work are as follows. First, the independence assumption in Theorem \ref{thm:wavem} can be relaxed by introducing in the RIF regression a control function constructed using instrumental variables to account for endogenous covariates. \redd{Second, the extension of this result to the case of heteroskedastic errors would be an interesting topic for future research.} %Third, the linear model considered for the conditional expectation of the RIF can be generalized using nonparametric regression approaches that does not impose functional form restrictions on $m_{\bold \theta_{\tau, \bs{u}}} (\bs X)$, such as polynomial or (B-)spline regressions. 
 Lastly, it would also be interesting the study of the theoretical behavior of the multivariate M-quantile cross-validated estimator when selecting the tuning constant of the Huber's influence function.

\bigskip
\begin{center}
{\large\bf SUPPLEMENTARY MATERIALS}
\end{center}

\begin{description}
%The supplementary material to this article contains the
\item[Simulations, additional results and proofs:] Simulation studies, additional results and technical derivations that are used to support the results in the manuscript. (PDF file)
%In an online appendix we provide details for the permutation test of the null hypothesis of no network dependence (A1), derive the moments of Moran?s I under the null (A2), provide further details of our simulations (A3), and provide further details of our analysis of FHS data (A4).
%The online supplement to this article contains a month-by-month assessment of the predictive performance for the marginal models assessed in Section 4.2, as well as details on computation time.
%\item[Title:] Brief description. (file type)
\end{description}

\bigskip
\begin{center}
{\large\bf FUNDING}
\end{center}

\begin{description}
\item[Competing interests:] The authors report there are no competing interests to declare.
\end{description}
%\cleardoublepage
%\section*{References}
%\bibliographystyle{elsarticle-harv} 
%\bibliographystyle{agsm}
\bibliographystyle{Chicago}
%\bibliographystyle{authoryear}
%\bibliographystyle{apalike}

%\scriptsize
\bibliography{biblio_manuscript}
\clearpage

\pagestyle{empty}
\begin{table}[!h]
% ln estimates
\centering
 \smallskip 
 %\footnotesize
 %\resizebox{15cm}{!} {
 %\scalebox{0.4} {
 \resizebox{0.85\columnwidth}{!}{
\begin{tabular}{lcccccccccccc}
\toprule
 & \multicolumn{6}{c}{Unconditional Regression} & \multicolumn{6}{c}{Conditional Regression} \\\cmidrule(r){2-7} \cmidrule(r){8-13}
$\tau$ & \multicolumn{2}{c}{0.10} & \multicolumn{2}{c}{0.50} & \multicolumn{2}{c}{0.90} & \multicolumn{2}{c}{0.10} & \multicolumn{2}{c}{0.50} & \multicolumn{2}{c}{0.90} \\\cmidrule(r){2-3} \cmidrule(r){4-5} \cmidrule(r){6-7} \cmidrule(r){8-9} \cmidrule(r){10-11} \cmidrule(r){12-13}
%Variable & Wealth & Consumption & Wealth & Consumption & Wealth & Consumption & Wealth & Consumption & Wealth & Consumption & Wealth & Consumption & Wealth & Consumption\\
Variable & W & C & W & C & W & C & W & C & W & C & W & C\\
%& & 10-${th}$ & 50-${th}$ & 90-${th}$ & 10-${th}$ & 50-${th}$ & 90-${th}$ \\
\hline
Intercept         & $\mathbf{-29.370}$ & $\mathbf{5.960}$  & $\mathbf{-5.701}$ & $\mathbf{5.472}$  & $-0.514$          & $\mathbf{4.742}$  & $\mathbf{-15.840}$ & $\mathbf{4.348}$  & $\mathbf{-5.701}$ & $\mathbf{5.472}$  & $0.003$           & $\mathbf{6.156}$  \\
                  & $(1.243)$          & $(0.159)$         & $(0.390)$         & $(0.095)$         & $(0.302)$         & $(0.188)$         & $(0.529)$          & $(0.121)$         & $(0.389)$         & $(0.094)$         & $(0.297)$         & $(0.122)$         \\
LINC              & $\mathbf{2.976}$   & $\mathbf{0.410}$  & $\mathbf{1.315}$  & $\mathbf{0.435}$  & $\mathbf{0.993}$  & $\mathbf{0.496}$  & $\mathbf{2.156}$   & $\mathbf{0.547}$  & $\mathbf{1.315}$  & $\mathbf{0.435}$  & $\mathbf{0.914}$  & $\mathbf{0.361}$  \\
                  & $(0.105)$          & $(0.016)$         & $(0.032)$         & $(0.008)$         & $(0.025)$         & $(0.015)$         & $(0.044)$          & $(0.010)$         & $(0.032)$         & $(0.008)$         & $(0.025)$         & $(0.010)$         \\
Sex               & $0.044$            & $\mathbf{-0.035}$ & $0.008$           & $\mathbf{-0.028}$ & $-0.001$          & $-0.011$          & $-0.022$           & $\mathbf{-0.037}$ & $0.008$           & $\mathbf{-0.028}$ & $0.016$           & $\mathbf{-0.030}$ \\
                  & $(0.134)$          & $(0.016)$         & $(0.045)$         & $(0.011)$         & $(0.033)$         & $(0.017)$         & $(0.061)$          & $(0.014)$         & $(0.045)$         & $(0.011)$         & $(0.034)$         & $(0.014)$         \\
Age               & $\mathbf{0.219}$   & $\mathbf{-0.014}$ & $\mathbf{0.081}$  & $0.001$           & $\mathbf{0.055}$  & $\mathbf{0.012}$  & $\mathbf{0.096}$   & $-0.003$          & $\mathbf{0.081}$  & $0.001$           & $\mathbf{0.058}$  & $\mathbf{0.007}$  \\
                  & $(0.025)$          & $(0.003)$         & $(0.008)$         & $(0.002)$         & $(0.006)$         & $(0.003)$         & $(0.011)$          & $(0.003)$         & $(0.008)$         & $(0.002)$         & $(0.006)$         & $(0.003)$         \\
Age2              & $\mathbf{-0.002}$  & $\mathbf{0.000}$  & $\mathbf{-0.001}$ & $-0.000$          & $\mathbf{-0.000}$ & $\mathbf{-0.000}$ & $\mathbf{-0.001}$  & $0.000$           & $\mathbf{-0.001}$ & $-0.000$          & $\mathbf{-0.000}$ & $\mathbf{-0.000}$ \\
                  & $(0.000)$          & $(0.000)$         & $(0.000)$         & $(0.000)$         & $(0.000)$         & $(0.000)$         & $(0.000)$          & $(0.000)$         & $(0.000)$         & $(0.000)$         & $(0.000)$         & $(0.000)$         \\
Marital status \\
\quad never married     & $\mathbf{0.523}$   & $\mathbf{-0.130}$ & $0.094$           & $\mathbf{-0.152}$ & $0.006$           & $\mathbf{-0.158}$ & $\mathbf{0.348}$   & $\mathbf{-0.101}$ & $0.094$           & $\mathbf{-0.152}$ & $0.010$           & $\mathbf{-0.175}$ \\
                  & $(0.170)$          & $(0.023)$         & $(0.056)$         & $(0.014)$         & $(0.042)$         & $(0.021)$         & $(0.077)$          & $(0.018)$         & $(0.056)$         & $(0.014)$         & $(0.043)$         & $(0.018)$         \\
\quad separated         & $\mathbf{-0.845}$  & $\mathbf{-0.056}$ & $\mathbf{-0.303}$ & $\mathbf{-0.104}$ & $\mathbf{-0.189}$ & $\mathbf{-0.141}$ & $\mathbf{-0.250}$  & $-0.041$          & $\mathbf{-0.303}$ & $\mathbf{-0.104}$ & $\mathbf{-0.239}$ & $\mathbf{-0.156}$ \\
                  & $(0.214)$          & $(0.026)$         & $(0.071)$         & $(0.017)$         & $(0.053)$         & $(0.027)$         & $(0.097)$          & $(0.022)$         & $(0.071)$         & $(0.017)$         & $(0.054)$         & $(0.022)$         \\
\quad widowed           & $\mathbf{0.408}$   & $\mathbf{-0.098}$ & $0.067$           & $\mathbf{-0.098}$ & $-0.008$          & $\mathbf{-0.084}$ & $\mathbf{0.323}$   & $-0.029$          & $0.067$           & $\mathbf{-0.098}$ & $-0.032$          & $\mathbf{-0.135}$ \\
                  & $(0.191)$          & $(0.025)$         & $(0.063)$         & $(0.015)$         & $(0.047)$         & $(0.024)$         & $(0.086)$          & $(0.020)$         & $(0.063)$         & $(0.015)$         & $(0.048)$         & $(0.020)$         \\
Education level \\
\quad middle school     & $\mathbf{0.434}$   & $\mathbf{0.125}$  & $\mathbf{0.244}$  & $\mathbf{0.102}$  & $\mathbf{0.173}$  & $\mathbf{0.059}$  & $\mathbf{0.274}$   & $\mathbf{0.085}$  & $\mathbf{0.244}$  & $\mathbf{0.102}$  & $\mathbf{0.204}$  & $\mathbf{0.118}$  \\
                  & $(0.174)$          & $(0.021)$         & $(0.058)$         & $(0.014)$         & $(0.043)$         & $(0.022)$         & $(0.079)$          & $(0.018)$         & $(0.058)$         & $(0.014)$         & $(0.044)$         & $(0.018)$         \\
\quad vocational school & $0.398$            & $\mathbf{0.133}$  & $\mathbf{0.299}$  & $\mathbf{0.117}$  & $\mathbf{0.278}$  & $\mathbf{0.065}$  & $\mathbf{0.303}$   & $\mathbf{0.064}$  & $\mathbf{0.299}$  & $\mathbf{0.117}$  & $\mathbf{0.296}$  & $\mathbf{0.165}$  \\
                  & $(0.242)$          & $(0.029)$         & $(0.081)$         & $(0.020)$         & $(0.060)$         & $(0.031)$         & $(0.110)$          & $(0.025)$         & $(0.081)$         & $(0.020)$         & $(0.062)$         & $(0.025)$         \\
\quad high school       & $\mathbf{1.006}$   & $\mathbf{0.156}$  & $\mathbf{0.611}$  & $\mathbf{0.189}$  & $\mathbf{0.510}$  & $\mathbf{0.191}$  & $\mathbf{0.550}$   & $\mathbf{0.128}$  & $\mathbf{0.611}$  & $\mathbf{0.189}$  & $\mathbf{0.548}$  & $\mathbf{0.246}$  \\
                  & $(0.189)$          & $(0.024)$         & $(0.063)$         & $(0.015)$         & $(0.047)$         & $(0.024)$         & $(0.086)$          & $(0.020)$         & $(0.063)$         & $(0.015)$         & $(0.048)$         & $(0.020)$         \\
\quad university        & $\mathbf{0.518}$   & $\mathbf{0.126}$  & $\mathbf{0.656}$  & $\mathbf{0.276}$  & $\mathbf{0.849}$  & $\mathbf{0.440}$  & $\mathbf{0.429}$   & $\mathbf{0.169}$  & $\mathbf{0.656}$  & $\mathbf{0.276}$  & $\mathbf{0.670}$  & $\mathbf{0.379}$  \\
                  & $(0.229)$          & $(0.028)$         & $(0.076)$         & $(0.019)$         & $(0.057)$         & $(0.030)$         & $(0.104)$          & $(0.024)$         & $(0.076)$         & $(0.019)$         & $(0.058)$         & $(0.024)$         \\
Employment status \\
\quad self-employed     & $\mathbf{1.791}$   & $\mathbf{-0.071}$ & $\mathbf{0.942}$  & $0.032$           & $\mathbf{0.947}$  & $\mathbf{0.153}$  & $\mathbf{0.842}$   & $-0.028$          & $\mathbf{0.942}$  & $0.032$           & $\mathbf{0.812}$  & $\mathbf{0.066}$  \\
                  & $(0.206)$          & $(0.025)$         & $(0.068)$         & $(0.017)$         & $(0.051)$         & $(0.027)$         & $(0.093)$          & $(0.021)$         & $(0.068)$         & $(0.017)$         & $(0.052)$         & $(0.021)$         \\
not-employed      & $\mathbf{1.524}$   & $-0.001$          & $\mathbf{0.654}$  & $\mathbf{0.041}$  & $\mathbf{0.575}$  & $\mathbf{0.111}$  & $\mathbf{0.658}$   & $0.021$           & $\mathbf{0.654}$  & $\mathbf{0.041}$  & $\mathbf{0.582}$  & $\mathbf{0.045}$  \\
                  & $(0.180)$          & $(0.023)$         & $(0.060)$         & $(0.015)$         & $(0.045)$         & $(0.023)$         & $(0.082)$          & $(0.019)$         & $(0.060)$         & $(0.015)$         & $(0.046)$         & $(0.019)$         \\
Geographical area \\
\quad centre            & $\mathbf{0.681}$   & $\mathbf{0.046}$  & $\mathbf{0.213}$  & $\mathbf{0.024}$  & $\mathbf{0.080}$  & $-0.001$          & $\mathbf{0.325}$   & $\mathbf{0.056}$  & $\mathbf{0.213}$  & $\mathbf{0.024}$  & $\mathbf{0.107}$  & $0.005$           \\
                  & $(0.143)$          & $(0.017)$         & $(0.048)$         & $(0.012)$         & $(0.036)$         & $(0.018)$         & $(0.065)$          & $(0.015)$         & $(0.048)$         & $(0.012)$         & $(0.036)$         & $(0.015)$         \\
\quad south and islands & $\mathbf{1.015}$   & $\mathbf{-0.073}$ & $\mathbf{0.196}$  & $\mathbf{-0.102}$ & $0.017$           & $\mathbf{-0.100}$ & $\mathbf{0.359}$   & $\mathbf{-0.055}$ & $\mathbf{0.196}$  & $\mathbf{-0.102}$ & $0.055$           & $\mathbf{-0.154}$ \\
                  & $(0.133)$          & $(0.019)$         & $(0.044)$         & $(0.011)$         & $(0.033)$         & $(0.017)$         & $(0.060)$          & $(0.014)$         & $(0.044)$         & $(0.011)$         & $(0.034)$         & $(0.014)$         \\
\\
$r_{12}$ & \multicolumn{2}{c}{0.324} & \multicolumn{2}{c}{0.473} & \multicolumn{2}{c}{0.653} \\
%cor RIF 0.6643777 0.4764579 0.4156240
%Huberised residuals & 95.36075139  0.07115412  0.00000000 & cond 86.3811015  0.1707699  0.0000000
\bottomrule
\end{tabular}}
%$u = (\frac{1}{2}, - \frac{1}{2})$
\caption{\footnotesize Unconditional and conditional regression coefficient estimates \redd{obtained from the linear specification for the RIF} at the investigated $\tau$ levels and direction $\bs u = (\frac{1}{\sqrt{2}},\frac{1}{\sqrt{2}})'$, using the optimal tuning constant $c = c^\star$. Parameter estimates are displayed in boldface when significant at the standard 5\% level.}\label{tab:coefmcstar}
\end{table}

\begin{table}[!h]
% ln estimates
\centering
 \smallskip 
 %\footnotesize
 %\resizebox{15cm}{!} {
 %\scalebox{0.4} {
 \resizebox{0.85\columnwidth}{!}{
\begin{tabular}{lcccccccccccc}
\toprule
 & \multicolumn{6}{c}{Unconditional Regression} & \multicolumn{6}{c}{Conditional Regression} \\\cmidrule(r){2-7} \cmidrule(r){8-13}
$\tau$ & \multicolumn{2}{c}{0.10} & \multicolumn{2}{c}{0.50} & \multicolumn{2}{c}{0.90} & \multicolumn{2}{c}{0.10} & \multicolumn{2}{c}{0.50} & \multicolumn{2}{c}{0.90} \\\cmidrule(r){2-3} \cmidrule(r){4-5} \cmidrule(r){6-7} \cmidrule(r){8-9} \cmidrule(r){10-11} \cmidrule(r){12-13}
%Variable & Wealth & Consumption & Wealth & Consumption & Wealth & Consumption & Wealth & Consumption & Wealth & Consumption & Wealth & Consumption & Wealth & Consumption\\
Variable & W & C & W & C & W & C & W & C & W & C & W & C\\
\hline
Intercept         & $-0.107$          & $\mathbf{11.293}$ & $\mathbf{7.275}$ & $\mathbf{9.798}$  & $\mathbf{7.230}$ & $\mathbf{7.708}$  & $\mathbf{5.450}$  & $\mathbf{9.264}$  & $\mathbf{7.275}$ & $\mathbf{9.798}$  & $\mathbf{6.342}$  & $\mathbf{9.943}$  \\
                  & $(6.859)$         & $(1.449)$         & $(1.669)$        & $(0.443)$         & $(0.874)$        & $(1.060)$         & $(1.669)$         & $(0.443)$         & $(1.669)$        & $(0.443)$         & $(0.766)$         & $(0.401)$         \\
LINC              & $\mathbf{3.383}$  & $\mathbf{0.463}$  & $\mathbf{1.662}$ & $\mathbf{0.587}$  & $\mathbf{1.416}$ & $\mathbf{0.765}$  & $\mathbf{2.254}$  & $\mathbf{0.551}$  & $\mathbf{1.662}$ & $\mathbf{0.587}$  & $\mathbf{1.340}$  & $\mathbf{0.639}$  \\
                  & $(0.343)$         & $(0.076)$         & $(0.083)$        & $(0.034)$         & $(0.182)$        & $(0.140)$         & $(0.083)$         & $(0.034)$         & $(0.083)$        & $(0.034)$         & $(0.089)$         & $(0.047)$         \\
Sex               & $0.114$           & $-0.017$          & $0.028$          & $\mathbf{-0.023}$ & $-0.006$         & $-0.021$          & $-0.006$          & $\mathbf{-0.039}$ & $0.028$          & $\mathbf{-0.023}$ & $0.042$           & $-0.015$          \\
                  & $(0.135)$         & $(0.016)$         & $(0.045)$        & $(0.011)$         & $(0.033)$        & $(0.016)$         & $(0.045)$         & $(0.011)$         & $(0.045)$        & $(0.011)$         & $(0.037)$         & $(0.013)$         \\
Age               & $\mathbf{0.060}$  & $\mathbf{0.046}$  & $\mathbf{0.112}$ & $\mathbf{-0.007}$ & $\mathbf{0.096}$ & $\mathbf{-0.006}$ & $\mathbf{0.223}$  & $\mathbf{0.043}$  & $\mathbf{0.112}$ & $\mathbf{-0.007}$ & $\mathbf{-0.024}$ & $\mathbf{-0.045}$ \\
                  & $(0.002)$         & $(0.000)$         & $(0.001)$        & $(0.000)$         & $(0.000)$        & $(0.000)$         & $(0.002)$         & $(0.000)$         & $(0.002)$        & $(0.000)$         & $(0.002)$         & $(0.001)$         \\
Marital status \\
\quad never married     & $\mathbf{0.634}$  & $\mathbf{-0.123}$ & $\mathbf{0.202}$ & $\mathbf{-0.102}$ & $\mathbf{0.152}$ & $\mathbf{-0.059}$ & $\mathbf{0.342}$  & $\mathbf{-0.097}$ & $\mathbf{0.202}$ & $\mathbf{-0.102}$ & $\mathbf{0.119}$  & $\mathbf{-0.096}$ \\
                  & $(0.176)$         & $(0.022)$         & $(0.058)$        & $(0.013)$         & $(0.038)$        & $(0.018)$         & $(0.058)$         & $(0.013)$         & $(0.058)$        & $(0.013)$         & $(0.046)$         & $(0.014)$         \\
\quad separated         & $\mathbf{-0.676}$ & $-0.029$          & $-0.149$         & $\mathbf{-0.037}$ & $-0.003$         & $-0.027$          & $\mathbf{-0.193}$ & $\mathbf{-0.039}$ & $-0.149$         & $\mathbf{-0.037}$ & $-0.079$          & $\mathbf{-0.042}$ \\
                  & $(0.239)$         & $(0.028)$         & $(0.079)$        & $(0.017)$         & $(0.052)$        & $(0.023)$         & $(0.079)$         & $(0.017)$         & $(0.079)$        & $(0.017)$         & $(0.063)$         & $(0.019)$         \\
\quad widowed           & $\mathbf{0.520}$  & $\mathbf{-0.083}$ & $\mathbf{0.184}$ & $\mathbf{-0.046}$ & $\mathbf{0.144}$ & $0.012$           & $\mathbf{0.359}$  & $-0.022$          & $\mathbf{0.184}$ & $\mathbf{-0.046}$ & $0.089$           & $\mathbf{-0.052}$ \\
                  & $(0.191)$         & $(0.023)$         & $(0.061)$        & $(0.015)$         & $(0.044)$        & $(0.021)$         & $(0.061)$         & $(0.015)$         & $(0.061)$        & $(0.015)$         & $(0.047)$         & $(0.017)$         \\
Education level \\
\quad middle school     & $0.339$           & $\mathbf{0.109}$  & $\mathbf{0.193}$ & $\mathbf{0.080}$  & $\mathbf{0.128}$ & $\mathbf{0.031}$  & $\mathbf{0.216}$  & $\mathbf{0.084}$  & $\mathbf{0.193}$ & $\mathbf{0.080}$  & $\mathbf{0.151}$  & $\mathbf{0.077}$  \\
                  & $(0.192)$         & $(0.023)$         & $(0.062)$        & $(0.013)$         & $(0.040)$        & $(0.016)$         & $(0.062)$         & $(0.013)$         & $(0.062)$        & $(0.013)$         & $(0.046)$         & $(0.015)$         \\
\quad vocational school & $0.268$           & $\mathbf{0.111}$  & $\mathbf{0.234}$ & $\mathbf{0.091}$  & $\mathbf{0.227}$ & $0.036$           & $\mathbf{0.238}$  & $\mathbf{0.067}$  & $\mathbf{0.234}$ & $\mathbf{0.091}$  & $\mathbf{0.227}$  & $\mathbf{0.120}$  \\
                  & $(0.246)$         & $(0.029)$         & $(0.082)$        & $(0.019)$         & $(0.056)$        & $(0.026)$         & $(0.082)$         & $(0.019)$         & $(0.082)$        & $(0.019)$         & $(0.065)$         & $(0.024)$         \\
\quad high school       & $\mathbf{0.778}$  & $\mathbf{0.122}$  & $\mathbf{0.458}$ & $\mathbf{0.124}$  & $\mathbf{0.347}$ & $\mathbf{0.090}$  & $\mathbf{0.462}$  & $\mathbf{0.123}$  & $\mathbf{0.458}$ & $\mathbf{0.124}$  & $\mathbf{0.387}$  & $\mathbf{0.126}$  \\
                  & $(0.192)$         & $(0.022)$         & $(0.064)$        & $(0.015)$         & $(0.045)$        & $(0.021)$         & $(0.064)$         & $(0.015)$         & $(0.064)$        & $(0.015)$         & $(0.049)$         & $(0.017)$         \\
\quad university        & $0.350$           & $\mathbf{0.125}$  & $\mathbf{0.392}$ & $\mathbf{0.151}$  & $\mathbf{0.444}$ & $\mathbf{0.168}$  & $\mathbf{0.424}$  & $\mathbf{0.154}$  & $\mathbf{0.392}$ & $\mathbf{0.151}$  & $\mathbf{0.333}$  & $\mathbf{0.143}$  \\
                  & $(0.212)$         & $(0.027)$         & $(0.075)$        & $(0.019)$         & $(0.061)$        & $(0.032)$         & $(0.075)$         & $(0.019)$         & $(0.075)$        & $(0.019)$         & $(0.061)$         & $(0.023)$         \\
Employment status \\
\quad self-employed     & $\mathbf{1.974}$  & $-0.016$          & $\mathbf{0.889}$ & $-0.005$          & $\mathbf{0.748}$ & $0.007$           & $\mathbf{0.932}$  & $\mathbf{-0.038}$ & $\mathbf{0.889}$ & $-0.005$          & $\mathbf{0.794}$  & $0.038$           \\
                  & $(0.178)$         & $(0.024)$         & $(0.058)$        & $(0.018)$         & $(0.056)$        & $(0.030)$         & $(0.058)$         & $(0.018)$         & $(0.058)$        & $(0.018)$         & $(0.051)$         & $(0.020)$         \\
\quad not-employed      & $\mathbf{1.950}$  & $\mathbf{0.070}$  & $\mathbf{0.685}$ & $\mathbf{0.043}$  & $\mathbf{0.437}$ & $0.026$           & $\mathbf{0.676}$  & $0.032$           & $\mathbf{0.685}$ & $\mathbf{0.043}$  & $\mathbf{0.617}$  & $\mathbf{0.062}$  \\
                  & $(0.230)$         & $(0.026)$         & $(0.074)$        & $(0.016)$         & $(0.051)$        & $(0.024)$         & $(0.074)$         & $(0.016)$         & $(0.074)$        & $(0.016)$         & $(0.061)$         & $(0.019)$         \\
Geographical area \\
\quad centre            & $\mathbf{0.681}$  & $\mathbf{0.043}$  & $\mathbf{0.237}$ & $\mathbf{0.036}$  & $\mathbf{0.125}$ & $0.029$           & $\mathbf{0.301}$  & $\mathbf{0.059}$  & $\mathbf{0.237}$ & $\mathbf{0.036}$  & $\mathbf{0.148}$  & $\mathbf{0.027}$  \\
                  & $(0.134)$         & $(0.014)$         & $(0.044)$        & $(0.010)$         & $(0.034)$        & $(0.016)$         & $(0.044)$         & $(0.010)$         & $(0.044)$        & $(0.010)$         & $(0.035)$         & $(0.012)$         \\
\quad south and islands & $\mathbf{1.181}$  & $\mathbf{-0.045}$ & $\mathbf{0.299}$ & $\mathbf{-0.059}$ & $\mathbf{0.117}$ & $\mathbf{-0.040}$ & $\mathbf{0.353}$  & $\mathbf{-0.052}$ & $\mathbf{0.299}$ & $\mathbf{-0.059}$ & $\mathbf{0.194}$  & $\mathbf{-0.060}$ \\
                  & $(0.152)$         & $(0.019)$         & $(0.046)$        & $(0.010)$         & $(0.030)$        & $(0.014)$         & $(0.046)$         & $(0.010)$         & $(0.046)$        & $(0.010)$         & $(0.037)$         & $(0.012)$         \\
\bottomrule
\end{tabular}}
%$u = (\frac{1}{2}, - \frac{1}{2})$
\caption{\footnotesize \redd{Unconditional and conditional regression coefficient estimates obtained from the nonparametric method for the RIF at the investigated $\tau$ levels and direction $\bs u = (\frac{1}{\sqrt{2}},\frac{1}{\sqrt{2}})'$, using the optimal tuning constant $c = c^\star$. Standard errors are computed via nonparametric bootstrap using 1000 resamples and parameter estimates are displayed in boldface when significant at the standard 5\% level.}}\label{tab:coefmcstar_spline}
\end{table}

\end{document}